\documentclass[11pt]{article}
\usepackage{graphicx}
\usepackage[curve,matrix,arrow,color]{xy}
\usepackage[mathscr]{eucal}

\xyoption{line}

\usepackage{epsf}
\usepackage{amssymb}
\usepackage{amsthm}
\usepackage{amsmath}
\usepackage{epstopdf}
\usepackage{color}

\makeatletter
\@addtoreset{equation}{section}
\makeatother
\renewcommand{\theequation}{\thesection.\arabic{equation}}
\def\begeqar{\begin{eqnarray}}
\def\endeqar{\end{eqnarray}}
\def\begeq{\begin{equation}}
\def\endeq{\end{equation}}

\def\wgta#1#2#3#4{\hbox{\rlap{\lower.35cm\hbox{$#1$}}
\hskip.2cm\rlap{\raise.25cm\hbox{$#2$}}
\rlap{\vrule width1.3cm height.4pt}
\hskip.55cm\rlap{\lower.6cm\hbox{\vrule width.4pt height1.2cm}}
\hskip.15cm
\rlap{\raise.25cm\hbox{$#3$}}\hskip.25cm\lower.35cm\hbox{$#4$}\hskip.6cm}}

\def\wgtb#1#2#3#4{\hbox{\rlap{\raise.25cm\hbox{$#2$}}
\hskip.2cm\rlap{\lower.35cm\hbox{$#1$}}
\rlap{\vrule width1.3cm height.4pt}
\hskip.55cm\rlap{\lower.6cm\hbox{\vrule width.4pt height1.2cm}}
\hskip.15cm
\rlap{\lower.35cm\hbox{$#4$}}\hskip.25cm\raise.25cm\hbox{$#3$}\hskip.6cm}}

\def\begeqar{\begin{eqnarray}}
\def\endeqar{\end{eqnarray}}

\newcommand{\one}{\boldsymbol{1}}
\newcommand{\tensor}{\otimes}

\newcommand{\q}{\mathfrak{q}}
\newcommand{\ffrac}[2]{\mbox{\footnotesize$\displaystyle\frac{#1}{#2}$}}
\newcommand{\idem}{\boldsymbol{e}}
\newcommand{\UresSL}[1]{\overline{U}_{\q} s\ell(#1)}

\newcommand{\TL}[1]{TL_{#1}}
\newcommand{\JTL}[1]{JTL_{#1}}

\newcommand{\modd}{\,\mathrm{mod}\,}
\newcommand{\proj}{p}

\newcommand{\inn}[2]{\langle#1,#2\rangle}

\newcommand{\LQG}{U_{\q} s\ell(2)}
\newcommand{\LQGi}{U_{i} s\ell(2)}

\newcommand{\LQGodd}{U^{\text{odd}}_{\q} s\ell(2)}

\newcommand{\K}{\mathsf{K}}
\newcommand{\F}{\mathsf{F}}
\newcommand{\FF}[1]{\F_{#1}}
\newcommand{\f}{\mathsf{f}}

\newcommand{\E}{\mathsf{E}}
\newcommand{\Er}{\tilde{\E}}
\newcommand{\Fr}{\tilde{\F}}
\newcommand{\EE}[1]{\E_{#1}}
\newcommand{\h}{\mathsf{h}}
\newcommand{\e}{\mathsf{e}}
\newcommand{\te}[1]{\mathsf{e}_{#1}}
\newcommand{\tf}[1]{\mathsf{f}_{#1}}
\newcommand{\epr}{\te{0}}
\newcommand{\fpr}{\tf{0}}

\newcommand{\fpl}{\tf{+}}
\newcommand{\epl}{\te{+}}
\newcommand{\epo}{\epr}
\newcommand{\fpo}{\fpr}

\newcommand{\half}{%
  \mathchoice{\ffrac{1}{2}}{\frac{1}{2}}{\frac{1}{2}}{\frac{1}{2}}}
  
\newcommand{\ramond}{\mathrm{a.p.}}

\newcommand{\ferm}{\theta}
\newcommand{\fermd}{\theta^{\dagger}}

\newcommand{\sfermp}{\psi^{2}}
\newcommand{\sfermm}{\psi^{1}}
\newcommand{\bsfermp}{\bar{\psi}^{2}}
\newcommand{\bsfermm}{\bar{\psi}^{1}}
\newcommand{\phip}{\phi^{2}}
\newcommand{\phim}{\phi^{1}}

\newcommand{\Hilb}{\mathcal{H}}
\newcommand{\vectv}[1]{|#1\rangle}
\newcommand{\veven}[1]{|v^{\text{even}}\rangle}
\newcommand{\vodd}[1]{|v^{\text{odd}}\rangle}
\newcommand{\vacl}{\langle\text{vac}|}
\newcommand{\vacr}{|\text{vac}\rangle}
\newcommand{\vac}{\boldsymbol{\Omega}}
\newcommand{\lvac}{\boldsymbol{\omega}}

\newcommand{\oN}{\mathbb{N}}
\newcommand{\oC}{\mathbb{C}}
\newcommand{\oZ}{\mathbb{Z}}

\newcommand{\step}{\epsilon}
\newcommand{\Endo}{\mathrm{End}}
\newcommand{\VEnd}{\mathcal{E}}
\newcommand{\Hom}{\mathrm{Hom}}

\newcommand{\Vir}{\mathcal{V}}

\newcommand{\VirN}{\boldsymbol{\mathcal{V}}}
\newcommand{\SU}{s\ell}
\newcommand{\nSU}{\boldsymbol{s\ell}}
\newcommand{\gl}{g\ell}
\newcommand{\psl}{ps\ell}

\newcommand{\HommTL}{\mathrm{Hom}_{\rule{0pt}{6.5pt}%
{\TL{N}}}}

\newcommand{\chVv}{\Hilb_{N}}
\newcommand{\chV}{V}
\newcommand{\repgl}{\pi_{\gl}}
\newcommand{\repXX}{\pi_{\mathrm{XX}}}
\newcommand{\repQG}{\rho_{\gl}}
\newcommand{\cent}{\mathfrak{Z}}
\newcommand{\centTL}{\cent_{\TL{}}}
\newcommand{\centJTL}{\cent_{\JTL{}}}
\newcommand{\centVir}{\cent_{\VirN}}
\newcommand{\N}{\mathrm N}

\newcommand{\toppr}{\mathsf{t}}
\newcommand{\botpr}{\mathsf{b}}
\newcommand{\leftpr}{\mathsf{l}}
\newcommand{\rightpr}{\mathsf{r}}
\newcommand{\stprp}{\mathsf{x}}

\newcommand{\XX}{\mathsf{X}}
\newcommand{\PP}{\mathsf{P}}

\newcommand{\TLX}{d^0}

\newcommand{\IrrTL}[1]{(d^0_{#1})}
\newcommand{\PrTL}[1]{\mathscr{P}_{#1}}
\newcommand{\StTL}[1]{\mathscr{W}_{#1}}
\newcommand{\MJTL}[1]{\mathscr{M}_{#1}}
\newcommand{\NJTL}[1]{\mathscr{N}_{#1}}

\newcommand{\Approjmodbase}{C}

\newcommand{\Falg}{\mathcal{A}}
\newcommand{\spco}{\mathcal{C}}

\newtheorem{thm}[subsubsection]{Theorem}
\newtheorem{Lemma}[subsection]{Lemma}

\newtheorem{Cor}[subsection]{Corrolary}

\theoremstyle{definition}

\newtheorem{dfn}[subsubsection]{Definition}
\newtheorem{rem}[subsubsection]{Remark}

\begin{document}
\begin{center}

\Large{Continuum limit and 
  symmetries  of  the periodic $\gl(1|1)$ spin chain }
\vskip 1cm

{\large A.M. Gainutdinov\,$^{a}$, N.~Read\,$^{b}$, and
H.~Saleur\,$^{a,c}$}

\vspace{1.0cm}

{\sl\small $^a$  Institut de Physique Th\'eorique, CEA Saclay,\\ 
Gif Sur Yvette, 91191, France\\}
{\sl\small $^b$ Department of Physics, Yale University, P.O. Box 208120,\\ New Haven, Connecticut 06520-8120, USA\\}
{\sl\small $^c$ Department of Physics and Astronomy,
University of Southern California,\\
Los Angeles, CA 90089, USA\\}

\end{center}

\begin{abstract}

This paper is the first in a series devoted to the study of
logarithmic conformal field theories (LCFT) in the bulk. Building on
earlier work in the boundary case, our general strategy consists in
analyzing the algebraic properties of lattice regularizations (quantum
spin chains) of these theories. In the boundary case, a crucial step
was the identification of the space of states as a bimodule over the
Temperley--Lieb (TL) algebra and the quantum group $\LQG$.  The extension of
this analysis in the bulk case involves considerable difficulties,
since the $\LQG$ symmetry is partly lost, while the TL algebra is
replaced by a much richer version (the Jones--Temperley--Lieb - JTL -
algebra). Even the simplest case of the $\gl(1|1)$ spin chain --
corresponding to the $c=-2$ symplectic fermions theory in the
continuum limit -- presents very rich aspects, which we will discuss in
several papers.

In this first work, we focus on the symmetries of the spin chain, that
is, the centralizer of the JTL algebra in the alternating tensor product of
the $\gl(1|1)$ fundamental representation and its dual. We prove that
this centralizer is only a subalgebra of $\LQG$ at $\q=i$ that we dub
$\LQGodd$. We then begin the analysis of the continuum limit of the
JTL algebra: using general arguments about the regularization of the
stress energy-tensor, we identify families of JTL elements going over
to the Virasoro generators $L_n,\bar{L}_n$ in the continuum limit. We
then discuss the $\SU(2)$ symmetry of the (continuum limit) symplectic
fermions theory from the lattice and JTL point of view.

The analysis of the spin chain as a bimodule over $\LQGodd$ and
$\JTL{N}$ is  discussed in the second paper of this series.

\end{abstract}


\section{Introduction}

There are often striking similarities between the properties of (not
necessarily integrable) lattice models and their conformally invariant
continuum limit in two dimensions.  The origin -- and mathematically more precise formulation --  of these similarities  is partly understood,
and related with the presence of common algebraic structures such as
quantum groups centralizers~\cite{PasquierSaleur, [FGST4],BFGT}. Nevertheless, many features remain unexplored in this
field, chief among them the relation between representations of the
Virasoro algebra and various lattice objects -- Temperley Lieb
algebras, RSOS paths~\cite{Kyoto,Kaufmann,PearceFeverati}, {\it etc.}

The similarities between lattice models and conformal field theories
(CFT) can be a powerful -- albeit non rigorous yet -- tool to infer
the continuum limit of some models which are too hard to solve
analytically. This idea has been exploited recently to deepen our
understanding of Logarithmic CFTs (LCFTs).  Indeed, models based on
representations of associative algebras such as the Temperley--Lieb (TL)
algebra exhibit~\cite{ReadSaleur07-2}, from a representation theoretic
point of view, and in {\sl finite size}, strong similarities with the
chiral algebras in LCFTs. The structure of indecomposable modules and
fusion rules carried out sometimes with great difficulty in the
Virasoro setting~\cite{Rohsiepe, GK1, MatRid} can then be
predicted from a more manageable algebraic analysis of the lattice
models~\cite{Pearceetal, ReadSaleur07-1,RP1}.
A rigorous reformulation of the similarities in representation theories
for lattice and continuum sides requires some categorical statements
like equivalence of tensor categories. The tensor structure or fusion data on
the lattice part is essentially an induction (bi)functor associated
with two chains of arbitrary sizes joined by a common vertex. The 
 construction of direct limits of `tensor' categories of modules over the lattice
algebras, {\it e.g.}, TL-modules, should then give  the desired equivalence with a
tensor category of modules over the chiral algebra in the continuum limit.
 
It has also turned out that, beyond the abstract structure of
indecomposable modules, the matrix elements of Virasoro generators
themselves can also be obtained from the lattice models, although this
time an extrapolation to infinite sizes and restriction to low energy
part of the spectrum have to be
implemented~\cite{KooSaleur}. Indecomposability parameters
characterizing Virasoro action in large
families of boundary LCFTs have recently been obtained in this
fashion~\cite{DubJacSal2, VasJacSal}.

While the case of boundary LCFTs is thus slowly getting under control,
the understanding of the {\sl bulk} case remains in its infancy. The
main problem here, from the continuum point of view, is the expected
double indecomposability of the modules
over the product of the left and right Virasoro algebras. From the
lattice point of view, the necessarily periodic geometry of the model
leads to more complicated algebras~\cite{MartinSaleur1, MartinSaleur},
and to a more intricate role of the quantum
group~\cite{PasquierSaleur}, whose symmetry is partly lost. A relative
understanding of bulk LCFTs has only been gained in the rational
case~\cite{GabRun,GabRunW} based on chiral
W-algebras~\cite{[K-first],[FGST3]}, and also for Wess--Zumino models on
supergroups which, albeit very simple as far as LCFTs go, provide
interesting lessons on the coupling of left and right
sectors~\cite{SaleurSchomerus}. We are not aware of much other work in  this area, apart from~\cite{DoFlohr}, and the recent very interesting paper~\cite{Ridout12}. 

\medskip

The present paper is the first in our investigation of bulk LCFTs using lattice models and algebras. We shall mostly deal with 
super-spin chains, which are now well
understood in the open case~\cite{ReadSaleur07-2}, and whose spectrum in the bulk  was determined 
as early as 2001~\cite{ReadSaleur01}. This spectrum
exhibits intricate patterns such as conformal weights covering all the
rationals (modulo integers), and large degeneracies given by complicated,
arithmetic formulas. To understand these patterns, and to extract the
structure of the left-right Virasoro representations, what is required is a more
thorough study of the lattice algebras present in this case. While
difficult, this study should not be impossible, thanks in part to
recent progress on the side of mathematics~\cite{MartinSaleur1,MartinSaleur,Jones,GL,GL1,Ram}.

Before launching into abstract algebra, it seems important to gain a
 better understanding of the potentially simplest case, that is the
 closed $\gl(1|1)$ spin chain, whose continuum limit is expected to be
 described by the ubiquitous symplectic fermion
 theory~\cite{Kausch}. Our goal is to understand this case thoroughly,
 in order to delineate a general strategy which we will be able to
 extend to other situations -- such as the $\gl(2|1)$ spin chain -- in
 subsequent papers. Unfortunately, even the $\gl(1|1)$ case is rather
 complicated, and will occupy us for a while.

Recall that a fundamental technical step in our approach is to analyze
the Hilbert space of the system as a (bi)module over the two algebras
-- the algebra of hamiltonian densities which, for the models in~\cite{ReadSaleur01,ReadSaleur07-1} is the periodically extended
Temperley--Lieb algebra, and its centralizing symmetry algebra.  We
will restrict in this first paper to the analysis of the symmetries,
postponing the full bimodule discussion to our second
paper~\cite{GRS2}. We begin with definitions of our closed spin-chains
and their relations with XX spin-chains in Sec.~\ref{sec:prelim}
where we also recall  the continuum limit in the open case. In the closed $\gl(1|1)$ case we shall see in
Sec.~\ref{sec:symm-sp-ch} that the symmetry algebra is only a
subalgebra of the symmetry $\LQG$ of the boundary theory. The
resulting object for periodic
conditions -- called $\LQGodd$, with $\q=i$ below -- is realized
as a subalgebra in $\LQG$ which involves the use of the Lusztig limit
$\q\to i$ of particular polynomials of odd degree in the quantum group
generators while the $\SU(2)$ subalgebra is given by polynomials of even
degree and realizes the symmetry for antiperiodic conditions. More rigorous statements are presented
in~Thm.~\ref{Thm:centr-JTL-main} and Thm.~\ref{Thm:centr-JTL-anti}.

A crucial feature of the product
$\VirN(2)=\Vir(2)\boxtimes\overline{\Vir}(2)$ of left and right
Virasoro algebras that appear in the continuum limit symplectic
fermion theory is the presence of a global $\SU(2)$ symmetry (the
`symplectic' symmetry of the theory). It turns out however that the
lattice centralizer of JTL, $\LQGodd$, does not contain the subalgebra
$\SU(2)$. What happens to this `extra symmetry' in the continuum limit
will turn out to be a crucial aspect of the problem of connecting
algebraic features of the lattice models with those of LCFTs. To
understand this better, we spend some time in 
Sec.~\ref{sec:scaling} analyzing the scaling limit of the spin
chain.
Using
general ideas about the lattice version of the stress energy tensor,
we identify particular `local' elements in the JTL algebra (such as
the generators $e_i$, or the commutators $[e_i,e_{i+1}]$) whose
long wavelength Fourier modes have a well-defined convergence to the left and right
Virasoro modes $L_n$ and $\bar{L}_n$ in the logarithmic theory of
symplectic fermions at $c=-2$. The fate of the $\SU(2)$ symmetry in
the $\gl(1|1)$ case is then discussed in
Sec.~\ref{sec:symm-continuum}.

A note on style: some of the results below -- roughly, all that
concerns algebraic aspects of the finite dimensional spin chain, as
presented in Sec.~\ref{sec:symm-sp-ch} and the three appendices -- are rigorous, and
presented accordingly in the form of propositions, theorems, {\it etc.}
While we believe the rest of the paper could be turned into fully
rigorous statements (at the price of dwelling into analysis), we have
chosen not to do so, and to remain instead close to the style of
physics literature.

Finally, we note that a lattice model going over in the continuum
limit to symplectic fermions with periodic boundary conditions has
been studied from a related but different point of view in~\cite{Saleursusy,PearceRasmussenVillani}.


\subsection{Notations}
To help the reader navigate through this paper, we provide a partial list of notations (common to this paper and its sequels):

\begin{itemize}

\item[\mbox{}] $\TL{N}$ --- the (ordinary) Temperley--Lieb algebra,

\item[\mbox{}] $\TL{N}^a$ --- the periodic Temperley--Lieb algebra,

\item[\mbox{}] $\JTL{N}$ --- the Jones--Temperley--Lieb algebra,

\item[\mbox{}] $\centJTL$ --- the  centralizer of $\JTL{N}$,

\item[\mbox{}] $\repgl$ --- the spin-chain representation of $\JTL{N}$,

\item[\mbox{}] $\LQG$ --- the full quantum group,

\item[\mbox{}] $\E$, $\F$,  $\K^{\pm1}$ --- the standard quantum group generators,

\item[\mbox{}] $\e$, $\f$ --- the renormalized powers of the generators $\E$ and $\F$,

\item[\mbox{}] $\repQG$ --- the spin-chain representation of the quantum group $\LQG$,

\item[\mbox{}] $\Vir(2)$ --- the left Virasoro algebra with $c=-2$,

\item[\mbox{}] $\VirN(2)$ --- the product of the left and right Virasoro
  algebras,

\item[\mbox{}] $\centVir$ --- the  centralizer of $\VirN(2)$,

\item[\mbox{}] $\nSU(2)$ --- Kausch's $\SU(2)$ symmetry.
\end{itemize}

\medskip

\section{Preliminaries}\label{sec:prelim}

 \subsection{The $\gl(1|1)$ super-spin chain}\label{sec:super-spin-ch-def}
 
The  $\gl(1|1)$ super-spin chain~\cite{ReadSaleur07-1} is the tensor product $\chVv=\tensor_{j=1}^{N}\chV_j$, with $\chV_j\cong\oC^2$, which consists of  $N=2L$ sites
labelled by $j=1,\ldots,2L$, with the fundamental representation of $\gl(1|1)$ on even sites and its dual on odd sites.
The  $\gl(1|1)$ algebra admits 
a free fermion representation based on operators $f_j$ and $f_j^\dagger$ which obey the anti-commutation relations 
\begin{eqnarray}\label{f-ferm-rel}
\{f_j,f_{j'}\}=0,\quad\{f^\dagger_j,f^\dagger_{j'}\}=0,\quad\{f_j,f_{j'}^\dagger\}=(-1)^{j}\delta_{jj'}.
\end{eqnarray}
The most general nearest-neighbour `Heisenberg' coupling
\begin{equation}\label{rep-JTL-1}
e_j^{\gl}= (f_j+f_{j+1})(f_j^\dagger+f_{j+1}^\dagger),\qquad 1\leq
j\leq N-1,
\end{equation}
is then a mapping onto  the $\gl(1|1)$-invariant in the
product of two neighbour tensorands\footnote{Note that this mapping is not a projector, as its square is equal to zero.}. It  can be expressed in terms of a
representation of the Temperley--Lieb algebra $\TL{2L}(m)$
 generated by $e_j$'s together with the identity, subject to the usual relations
\begin{eqnarray}
e_j^2&=&me_j,\nonumber\\
e_je_{j\pm 1}e_j&=&e_j,\label{TL}\\
e_je_k&=&e_ke_j\qquad(j\neq k,~k\pm 1),\nonumber
\end{eqnarray}
where $j=1,\ldots,N-1$. The operators
$e_j^{\gl}$ in~\eqref{rep-JTL-1}  satisfy the Temperley--Lieb algebra relations
with $m=0$ (in general for the models of~\cite{ReadSaleur07-1}, the parameter $m$ is
the superdimension of the fundamental representation). 
The open $\gl(1|1)$ spin-chain described by the coupling~\eqref{rep-JTL-1} and the Hamiltonian $-\sum_{j=1}^{N-1}e_j^{\gl}$ provides a faithful representation of $\TL{2L}(0)$.

 The closed (periodic) spin-chain is obtained simply by adding a coupling between
the sites with $j=2L$ and $j=1$, that is by adding a generator 
\begin{equation}\label{rep-JTL-2}
e_{2L}^{\gl}=(f_{2L}+f_{1})(f_{2L}^\dagger+f_{1}^\dagger),
\end{equation}
which corresponds to the periodic boundary condition
$f^{(\dagger)}_{2L+1}=f^{(\dagger)}_1$ on the lattice fermions, where notation such as $f^{(\dagger)}$ means the result holds both for $f$ and for $f^\dagger$.
The operators $e_j^{\gl}$, with $1\leq j\leq 2L$, satisfy the
relations~\eqref{TL} with $m=0$ where the indices are now  interpreted
modulo $N$ (the
abstract algebra generated by $e_j$ with these relations as the
defining relations is a quotient of the affine Hecke algebra of $A$-type and
is also known as the periodic
Temperley--Lieb algebra~\cite{MartinSaleur1,MartinSaleur}.) Note that all the operators  $e_j^{\gl}$ are self-adjoint with respect to the non-degenerate inner product
 (defined such that $\inn{f_j x}{y}=\inn{x}{f_j^{\dagger}y}$ for any $x,y\in\Hilb_{2L}$), which is indefinite due to the sign factor in~\eqref{f-ferm-rel}. 

  The critical Hamiltonian for our model is then expressed as
 \begin{equation}\label{Hamilt-def}
 H=-\sum_{j=1}^{2L}e^{\gl}_j
 \end{equation}
  (note that for this model the sign of $H$ is irrelevant, as the
algebra obeyed by $e_j$'s and $-e_j$'s are identical. This is of
course not the case for other values of $m$). We note that the Hamiltonian is also self-adjoint.

In the periodic case, we also  consider the
generators $u^2$ and $u^{-2}$ of translations by two sites to the right
and to the left, respectively.  The following additional 
relations are then obeyed,
\begin{equation}\label{aTL-rel}
\begin{split}
u^2e_ju^{-2}&=e_{j+2},\\
u^2e_{N-1}&=e_1\ldots e_{N-1}.
\end{split}
\end{equation}
%
The expressions for the $e_j^{\gl}$ defined in~\eqref{rep-JTL-1} and \eqref{rep-JTL-2}
 together with the translations $u^{\pm2}$ of the periodic spin-chain
provides a  representation of the so-called  \textit{Jones--Temperley--Lieb} (JTL)
algebra $\JTL{2L}(m=0)$ which we denote by
$\repgl:\JTL{2L}(0)\to\Endo_{\oC}(\chVv)$. The representation $\repgl$
is known to be non-faithful and non-semisimple~\cite{ReadSaleur07-1}. We give a precise definition of the JTL algebra in our second paper~\cite{GRS2}. 
In the following, we usually  suppress all
 reference to $m$ and suppose $m=0$. 

 
\subsection{A relation with XX spin-chains}
It will be useful in what follows to observe that the $\gl(1|1)$
spin-chain representation $\repgl$ is equivalent to a twisted XX
spin-chain representation $\repXX$ of $\JTL{2L}$. The expression of
the Temperley--Lieb generators in this case is well known for the open
chain~\cite{PasquierSaleur},
\begin{equation}
\repXX(e_j) \equiv e^{XX}_j=-\half\left[\sigma_j^x\sigma_{j+1}^x+\sigma_j^y\sigma_{j+1}^y - i(\sigma_j^z-\sigma_{j+1}^z)\right],
\end{equation}
%
where $\sigma^x_j$, $\sigma^y_j$ and $\sigma^z_j$ are usual Pauli matrices  acting on a $j$th tensorand,
\begin{equation}\label{Pauli}
\sigma^x = 
\begin{pmatrix}
0 & 1\\
1 & 0
\end{pmatrix}, \quad
\sigma^y = 
\begin{pmatrix}
0 & -i\\
i & 0
\end{pmatrix}, \quad
\sigma^z = 
\begin{pmatrix}
1 & 0\\
0 & -1
\end{pmatrix}.
\end{equation}
We also use the notations $\sigma^{\pm}=\half\bigl(\sigma^x\pm i\sigma^y\bigr)$ in what follows.

 To get  equivalence  in the closed case we need to  set in the expression for $e^{XX}_{2L}$ the following:
\begin{equation}
\sigma_{2L+1}^{\pm}=-(-1)^{S^z}\sigma_1^\pm,\qquad \text{with}\quad S^z=\half\sum_{j=1}^{2L}\sigma_j^z.
\end{equation}
This means that a periodic $\gl(1|1)$ (alternating) spin-chain
corresponds to a periodic XX spin-chain for odd values of the spin
 $S^z$ and to an antiperiodic  XX spin chain 
for even values.
	      
To prove this -- and for later computational simplicity -- it is
 useful to reformulate everything in terms of ordinary fermions
 $c_j^{(\dagger)}$ obeying anticommutation relations
 $\{c^{(\dagger)}_{j},c^{(\dagger)}_{j'}\}=0$, $\{c_j,c_{j'}^\dagger\}=\delta_{jj'}$. Starting from the XX representation $\repXX$ and using the
 Jordan--Wigner transformation %
 \begin{eqnarray}\label{JW-trans}
 c_j^\dagger&=&i^{j-1}~i^{\sigma_1^z+\ldots+\sigma_{j-1}^z}\otimes{\sigma_j^+},\nonumber\\
  c_j&=
 &i^{-j+1}~i^{-\sigma_1^z-\ldots-\sigma_{j-1}^z}\otimes {\sigma_j^-}
 \end{eqnarray} 
(in each case, both $i$ and $-i=i^{-1}$ can indeed be used
 interchangeably, as the whole prefactor is real), one obtains
\begin{equation}\label{eq:PTL-rep-first}
e^{XX}_j=c_j c_{j+1}^\dagger + c_{j+1} c_j^\dagger 
  +i\left(c_j^\dagger c_j-c^\dagger_{j+1}c_{j+1}\right),
\quad c^{(\dagger)}_{2L+1}=(-1)^{L}c^{(\dagger)}_1,\qquad 1\leq j\leq 2L.
\end{equation}

Meanwhile, we can also  reexpress the $f_j^{(\dagger)}$'s  from the
 $\gl(1|1)$ chain in terms of  these ordinary fermions:
\begin{equation}\label{eq:def-f-c}
f_j^\dagger=i^{j}c_j^\dagger,\qquad
f_j=i^{j}c_j
\end{equation}
leading to the identification
\begin{equation}
e^{\gl}_j=i(-1)^j\left[c_{j+1}c_j^\dagger +c_jc_{j+1}^\dagger+i(c_j^\dagger c_j-c_{j+1}^\dagger c_{j+1})\right]
= i(-1)^je^{XX}_j
\end{equation}
which gives an isomorphism of $\repgl$ with the representation of $\JTL{2L}$~\eqref{eq:PTL-rep-first} obtained in 
 the XX chain (the factor $i(-1)^j$ leaving the cubic relation
 invariant). 
We note also that our periodic $\gl(1|1)$ chain corresponds to
periodic ordinary fermions $c_j^{(\dagger)}$ if $L$ is even, and antiperiodic fermions
if $L$ is odd.
 

\subsection{The continuum limit and the importance of the symmetry algebra. }
  The continuum limit of
the  $\gl(1|1)$ spin chain~\eqref{Hamilt-def} is well
known~\cite{Saleursusy,ReadSaleur07-1}, and 
  corresponds to the symplectic fermions
logarithmic CFT at $c=-2$~\cite{Kausch}. It also describes the long distance properties
of dense polymers. Less well known are the associated algebraic
features like lattice construction of left and right Virasoro modes
  $L_n$, $\bar{L}_n$ based on $\JTL{2L}$, as well as 
  the centralizer of $\JTL{2L}$, 
 which are the main topic of this paper.
We recall here briefly  that for  an algebra $A$ and its representation space
  $\chVv$, the centralizer of $A$ is an algebra
  $\cent_{A}$ of all commuting operators $[\cent_{A},A]=0$,
  {\it i.e.}, the centralizer is defined as the algebra of intertwiners $\cent_{A}=\Endo_{A}(\chVv)$.

In the open case, the $\gl(1|1)$ spin chain exhibits a large symmetry algebra
dubbed ${\cal A}_{1|1}$ in~\cite{ReadSaleur07-1}. This algebra is the
centralizer $\centTL$ of $\TL{2L}(0)$ and is
generated by the identity and the five generators
\begin{eqnarray}\label{FF-sym-def}
F_{(1)}&=&\sum_j f_j\nonumber,\\
F^\dagger_{(1)}&=&\sum_j f_j^\dagger\nonumber,\\
F_{(2)}&=&\sum_{j<j'}f_jf_{j'},\\
F_{(2)}^\dagger&=&\sum_{j<j'} f_{j'}^\dagger f_j^\dagger\nonumber,\\
{\N}&=&\sum_j (-1)^jf_j^\dagger f_j-L\nonumber,
\end{eqnarray}
where the fermions-number operator $\N$ should not be confused
with the notation for a number of sites $N$.
%
The operators $F_{(1)}$, $F_{(1)}^\dagger$ generate the subalgebra
$\psl(1|1)$ while $F_{(2)}$, $F^\dagger_{(2)}$, and $\N$ generate an
$s\ell(2)$ Lie subalgebra, with respect to which
$F_{(1)}$ and $F_{(1)}^\dagger$ transform as a doublet. The resulting Lie
superalgebra~${\cal A}_{1|1}$ is the semi-direct product of these two
algebras. It turns out to coincide with the full quantum group
representation $\repQG\bigl(\LQG\bigr)$, for
$\q=i$ (see Sec.~\ref{sec:symm-sp-ch} for  definitions).

\subsubsection{The continuum (scaling) limit}\label{prelim:scal-lim}
It is time here to discuss a bit more precisely what is meant by
the continuum limit, first in the general case.  It is always
possible~\cite{ReadSaleur07-1} to consider a
$N\to\infty$ limit (or so-called projective/inductive limit) of the algebraic structures in the spin-chains,
especially the centralizer of the TL algebra and its modules, and the
modules over the TL algebra as well, from a purely algebraic point of view.
But for our purpose more is required. We have chosen a
Hamiltonian $H$ for the spin-chain (such as~\eqref{Hamilt-def}), which is an element of (the representation of) an algebra like $\TL{N}$ or $\JTL{N}$ to which we refer to as the ``hamiltonian
densities'' algebra. Physically, we
 focus on low-energy (and long-wavelength) properties in a
$N\to\infty$ limit. We can for instance introduce a lattice spacing
 between sites and  consider the limit as taken with a
lattice spacing distance tending to zero as $N\to\infty$, such that
the length of the chain remains constant in the limit, equal to $1$,
say (hence the term continuum limit), and also with 
the Hamiltonian $H$ rescaled by $N$.
 Then, low energies and long wavelengths
mean excitation energies and wavevectors of order $1$ in these
units. We are especially interested in cases where this continuum
limit is a non-trivial conformal field theory, which in these units
implies that excited states at energies of order $1$ above the ground
state do exist.  Note that in practice, it is equivalent and more
convenient to keep the lattice spacing constant as $N\to\infty$. In
this case, low energies and long wavelengths mean excitation energies
and wavevectors of order $1/N$. To get finite results to be compared
with those of the CFT one must, for instance, rescale then the gaps by
$N$, hence the name scaling limit, which we will use equivalently.

 It is not entirely clear how the limit can be taken in a
mathematically rigorous way, but roughly we want to take the
eigenvectors of $H$ that have low-energy eigenvalues only, and we
expect that the inner products among these vectors can be made to tend
to some limits. Further, if we focus on long wavelength Fourier
components of the set of local generators of the hamiltonian densities
algebra, we expect their limits to exist, and their commutation
relations to tend to those of the Virasoro generators $L_n$ (or $L_n +
\bar{L}_{-n}$ in the closed chain case), in the sense of strong
convergence of operators in the basis of low-energy
eigenvectors\footnote{See a more precise reformulation in
Sec.~\ref{sec:emerge-latVir} in the case of the periodic $\gl(1|1)$
spin-chain.}. Then, the modules over the (J)TL algebra restricted to
the low-energy states become in the scaling limit modules over the
universal enveloping algebra of the Virasoro algebra (the product of
left and right Virasoro algebras in the closed chain case), or
possibly even a larger algebra.

\medskip
An advantage in using the centralizer  is that 
it gives a control on representation theory of the ``hamiltonian
densities'' algebra on a finite chain and even on fusion
rules, as was demonstrated in~\cite{ReadSaleur07-2}.
It is clear that the centralizer of the hamiltonian densities is a symmetry of the low-lying spectrum of the Hamiltonian for any finite $N$.
The symmetry (centralizer) algebra in the scaling limit, which commutes with the Virasoro algebra (the product of left and right Virasoro algebras in the closed chain case), must be thus at least as large as that in the finite-$N$ chains.
For example,
the decomposition of the open $\gl(1|1)$ spin-chain as a (bi)module over the pair
$(\TL{N},{\cal A}_{1|1})$ of mutual centralizers goes over in the
scaling limit to a semi-infinite (`staircase') (bi)module~\cite{ReadSaleur07-2} over the
Virasoro algebra $\Vir(2)$, with the central charge $c=-2$, and (the
scaling limit of) ${\cal A}_{1|1}$, which is just an infinite-dimensional representation
of $\LQGi$. In this case, we thus have essentially the same centralizer for lattice and continuum models.

\medskip

While the scenario described above can not be fully established
analytically for general models, it is confirmed a posteriori by the
validity of the results obtained in~\cite{ReadSaleur07-2}. Of course,
in some special cases such as free theories, much more can be said,
and we will go back to the question, and a more rigorous
reformulation, of the scaling limit for the closed $\gl(1|1)$
spin-chains and the associated symplectic fermions CFT in the
following sections.

 In the periodic $\gl(1|1)$ spin-chains, while the $\gl(1|1)$ symmetry remains, the
equivalent of the generators $F_{(2)}$ and $F^\dagger_{(2)}$ introduced in~\eqref{FF-sym-def} disappears,
since the summation, extended around the chain, vanishes by
anticommutation of the $f_j$'s. Meanwhile, the Temperley--Lieb algebra
is replaced by $\JTL{N}$.  What replaces the appealing symmetry algebra known to exist in the open case when one turns to periodic
systems is the subject of the following  section.

\section{Symmetries for the spin chain}\label{sec:symm-sp-ch}
 
 \subsection{Quantum group results}\label{sec:qg-res}
 
We find it convenient here to start with some notations and results
about quantum groups when the deformation parameter $\q$ is a root of
unity.  The \textit{full} quantum group $\LQG$ with $\q = e^{i\pi/p}$,
for integer $p\geq2$, is generated by $\E$, $\F$, $\K^{\pm1}$, and $\e$, $\f$,
$\h$. The first three generators satisfy the standard quantum-group relations
\begin{equation*}
  \K\E\K^{-1}=\q^2\E,\quad
  \K\F\K^{-1}=\q^{-2}\F,\quad
  [\E,\F]=\ffrac{\K-\K^{-1}}{\q-\q^{-1}},
\end{equation*}
with additional relations
\begin{equation*}
  \E^{p}=\F^{p}=0,\quad \K^{2p}=\one,
\end{equation*}
and  the divided powers $\f\sim \F^p/[p]!$ and $\e\sim \E^p/[p]!$ satisfy the usual $s\ell(2)$-relations:
\begin{equation*}
  [\h,\e]=\e,\qquad[\h,\f]=-\f,\qquad[\e,\f]=2\h.
\end{equation*}
The full list of relations with comultiplication formulae are borrowed
from~\cite{BFGT} and listed in App.~A where we also give the simple correspondence with the 
quantum group generators
$S^{\pm}$, $S^z$ and $\q^{S^z}$ used commonly in the spin chain literature.

 For applications to $\gl(1|1)$
spin-chains, we consider only the case $p=2$ and set  in what follows
$\q\equiv i$. As a module over $\LQG$, the spin
chain $\chVv$ is a tensor product of
two-dimensional irreducibe representations such that $\E\to\sigma^+$,
$\F\to\sigma^-$, $\K\to \q\sigma^z$, and $\e=\f=0$, where $\sigma^{\pm}=\half\bigl(\sigma^x\pm i\sigma^y\bigr)$ and the Pauli matrices $\sigma^{x,y,z}$ are from~\eqref{Pauli}.
Using the $(N-1)$-folded
comultiplications~\eqref{N-fold-comult-cap},
\eqref{N-fold-comult-ren-e}, and~\eqref{N-fold-comult-ren-f} together
with the Jordan-Wigner transformation~\eqref{JW-trans}, we obtain
the representation $\repQG:\LQG\to\Endo_{\oC}(\chVv)$ (usual fermionic expressions)
\begin{align}
&\repQG(\E)\equiv\Delta^{N-1}(\E) = \sum_{1\leq j\leq N}\q^j c_j^{\dagger}\, \repQG(\K) = F^\dagger_{(1)}\,\repQG(\K),\notag\\
 &\repQG(\F)\equiv\Delta^{N-1}(\F) = \sum_{1\leq j\leq N}\q^{j-1} c_j = \q^{-1}F_{(1)},\label{QG-fermf-1}\\
&\repQG(\K)\equiv\Delta^{N-1}(\K) = (-1)^{\repQG(2\h)},\notag
\end{align}
and
\begin{eqnarray}
\repQG(\e)&\equiv&\Delta^{N-1}(\e) = \sum_{1\leq j_1<j_2\leq N}(-1)^{j_1+j_2}\q^{1-j_1-j_2} c_{j_1}^{\dagger}c_{j_2}^{\dagger}=\q^{-1}\sum_{1\leq j_1<j_2\leq N} f_{j_1}^\dagger f_{j_2}^\dagger=\q^{-1} F_{(2)}^\dagger,\nonumber\\
\repQG(\f)&\equiv&\Delta^{N-1}(\f) = \sum_{1\leq j_1<j_2\leq N}\q^{j_1+j_2-1} c_{j_1}c_{j_2} =\q\sum_{1\leq j_1<j_2\leq N} f_{j_1}f_{j_2}=\q F_{(2)},\label{QG-fermf-2}\\
\repQG(2\h)&\equiv&[\repQG(\e),\repQG(\f)]=\sum_{1\leq j\leq N}(-1)^j f_j^\dagger f_j-L,\nonumber
\end{eqnarray}
where we also detailed the correspondence with the generators~\eqref{FF-sym-def} of the $\TL{}$-centralizer $\centTL={\cal A}_{1|1}$.

\newcommand{\egl}{e^{g\ell}}

As  noted above, the symmetry algebra ${\cal A}_{1|1}$ of the
open spin-chain~\cite{ReadSaleur07-1}
coincides with the representation of the full quantum group $\repQG\bigl(\LQG\bigr)$, for
$\q=i$. The $\gl(1|1)$ (in fact $\psl(1|1)$ completed with
$(-1)^{\N}$) meanwhile corresponds to the representation of the \textit{restricted} quantum
group $\UresSL 2$ generated by $\E$, $\F$, and $\K^{\pm 1}$, with
$\E:\Hilb_{[n]}\to\Hilb_{[n+1]}$ and $\F:\Hilb_{[n]}\to\Hilb_{[n-1]}$
(satisfying  $\F^2=\E^2=0$)
 and $\Hilb_{[n]}$ denotes the subspace with $2\h=S^z=n$. The
statement that the representation $\repgl$ of the $\JTL{N}$ algebra obtained
from the periodic $\gl(1|1)$
spin-chain~\eqref{rep-JTL-1}-\eqref{Hamilt-def} does exhibit the $\gl(1|1)$ symmetry
corresponds to
an inclusion\footnote{A
  similar observation was
made in~\cite{PasquierSaleur} for $\egl_j$ replaced by $H^{1+
(S^z\,\text{mod}\,2)}$, where $H^{0,1}$ denotes the periodic
(resp. antiperiodic) XX spin-chain Hamiltonian.} $\repQG\bigl(\UresSL 2\bigr)\subset\centJTL$.  The question is whether there are more generators in the centralizer
  $\centJTL$ of $\JTL{N}$. 

\subsection{Fourier transforms}\label{sec:fourier}

It is convenient in the following to use Fourier transforms, and introduce, for $1\leq m\leq N$ (recall that we set $N=2L$),
\begin{equation}\label{def-ferm}
\ferm_{p_m} = \ffrac{1}{\sqrt{N}}\sum_{k=1}^N e^{-ikp_m}c_{k}, \qquad
\fermd_{p_m} = \ffrac{1}{\sqrt{N}}\sum_{k=1}^N e^{ikp_m}c^{\dagger}_{k}
\end{equation}
with the set of allowed momenta
\begin{equation}\label{momenta-set}
p_m=
\begin{cases}
\frac{2\pi m}{N},\qquad &L-\text{even},\\
\frac{(2m-1)\pi}{N}, &L-\text{odd},
\end{cases}
\qquad\quad 1\leq m\leq N,
\end{equation}
and with the usual anti-commutation relations
\begin{equation*}
\{\ferm_{p_1},\fermd_{p_2}\} = \delta_{p_1,p_2}, \qquad \{\ferm_{p_1},\ferm_{p_2}\} = \{\fermd_{p_1},\fermd_{p_2}\} = 0.
\end{equation*}

\subsubsection{Quantum group generators}
We then find using a direct calculation that 
\begin{equation}\label{QG-ferm-1}
\repQG(\E)=  \sqrt{N}\fermd_{\pi/2}(-1)^{S^z},
\qquad 
\repQG(\F) = \q^{-1}\sqrt{N}\ferm_{3\pi/2},
\end{equation}
and the renormalized powers read
\begin{equation}
\repQG(\e) =-\q\sum_{p\ne \frac{\pi}{2}}\frac{e^{i(\frac{\pi}{2}+p)}\fermd_p\,\fermd_{\pi-p} 
+ 2\,\fermd_p\,\fermd_{\pi/2}}{e^{i(\frac{\pi}{2}+p)}+1}
=\sum_{\substack{p=\frac{\pi}{2}+\frac{2\pi}{N}\\\text{step}=\frac{2\pi}{N}}}^{\frac{3\pi}{2}}
\tan\ffrac{1}{2}\bigl(\ffrac{\pi}{2}+p\bigr)\,\fermd_{p}\,\fermd_{\pi-p} 
- 2i\sum_{p\ne\frac{\pi}{2}}\frac{\fermd_p\,\fermd_{\pi/2}}{e^{i(\frac{\pi}{2}+p)}+1},\label{QG-ferm-2}
\end{equation}
and
\begin{equation}\label{QG-ferm-3}
\repQG(\f) =
\q\sum_{p\ne \frac{3\pi}{2}}\frac{e^{i(\frac{3\pi}{2}-p)}\ferm_p\,\ferm_{\pi-p} 
- 2\,\ferm_p\,\ferm_{3\pi/2}}{e^{i(\frac{3\pi}{2}-p)}-1}
=-\sum_{\substack{p=\frac{\pi}{2}\\\text{step}=\frac{2\pi}{N}}}^{\frac{3\pi}{2}-\frac{2\pi}{N}}
\cot\ffrac{1}{2}\bigl(\ffrac{\pi}{2}+p\bigr)\,\ferm_{p}\,\ferm_{\pi-p} 
- 2i\sum_{p\ne\frac{3\pi}{2}}\frac{\ferm_p\,\ferm_{3\pi/2}}{e^{-i(\frac{\pi}{2}+p)}-1}.
\end{equation}
These results agrees with ones established before in~\cite{DFMC} in a slightly different basis.

\subsubsection{JTL generators in terms of  Fourier transforms}

Finally, we can reexpress the generators $e_j$ of $\JTL{N}$ themselves:
\begin{equation}\label{eq:PTL-ferm}
e_j^{\gl}= (-1)^j \ffrac{i}{N}\sum_{p_1,p_2}e^{ij(p_2-p_1)}( i - e^{-ip_1})(1+ie^{ip_2})
\fermd_{p_1}\ferm_{p_2},\qquad 1\leq j\leq N,
\end{equation}
where the sum is taken over all possible momenta defined
in~\eqref{momenta-set}. In what follows, we use simply the notation
$e_j$ for the representation $e_j^{\gl}$ in~\eqref{eq:PTL-ferm}.
 
 In order to
translate (the sub-index of) the $\JTL{N}$ generators $e_j$, we demand
\begin{equation}\label{gen-u-acts}
u^2 f_j^{(\dagger)}u^{-2}=f_{j+2}^{(\dagger)}
\end{equation}
which means, in terms of the Fourier modes, that
\begin{equation}
u^2 \theta_{p_m}u^{-2}=-e^{2ip_m}  \theta_{p_m},
\qquad
u^2 \theta_{p_m}^{\dagger}u^{-2}=-e^{- 2ip_m}  \theta_{p_m}^{\dagger}.
\end{equation}
It is then convenient to express the generator $u^2$ in terms of these
Fourier modes. For this, we observe that, if $\theta$ and $\theta^\dagger$
are a conjugate pair of fermions, $\theta^2=
\left(\theta^\dagger\right)^2=0$, $\{\theta,\theta^\dagger\}=1$, we
have
\begin{eqnarray*}
e^{\lambda \theta^\dagger\theta}~\theta~ e^{-\lambda \theta^\dagger\theta}=e^{-\lambda} \theta\nonumber,\\
e^{\lambda \theta^\dagger\theta}~\theta^\dagger ~e^{-\lambda \theta^\dagger\theta}=e^{\lambda}\theta^\dagger,
\end{eqnarray*}
from which we can finally write the coherent state
representation
\begin{equation}\label{uu-theta}
u^2=\exp\left[\sum_{m=1}^{N} (i\pi-2i p_m)\theta_{p_m}^\dagger\theta_{p_m}\right].
\end{equation}

\medskip

We can then easily check the only linear combinations of fermions
which commute with $e_j$, where $1\leq j\leq N$, and $u^2$ are $\ferm_{3\pi/2}$ and
$\fermd_{\pi/2}$ . So, we have\footnote{We sometimes simplify
  expressions omitting more bulky and pedantic notations like $\bigl[\repgl(\JTL{N}),\repQG\bigl(\UresSL 2\bigr)\bigr]=0$.}
\begin{equation*}
\left[\JTL{N},\UresSL 2\right]=0,
\end{equation*}
as was mentioned above.
To find
additional generators in the centralizer~$\centJTL$, we look for 
 elements in the centralizer $\centTL$ of the subalgebra $\TL{N}\subset\JTL{N}$. This centralizer is the 
quantum group $\LQG$, which differs from 
$\UresSL2$)  by the presence of renormalized powers $\e$ and $\f$, and
the Cartan $\h=S^z/2$. It will turn out that the centralizer of $\JTL{N}$
can be  identified with the Lusztig limit ($\q\to i$) of appropriate
polynomials in generators of $\LQG$, as we now describe.

\subsection{The centralizer of $\JTL{N}$}

Using~\eqref{QG-ferm-2}, \eqref{QG-ferm-3}
and~\eqref{eq:PTL-ferm}, we calculate the commutators between $e_j$
and the renormalized powers,
\begin{align}
[e_j,\f] &= 
\begin{cases}
0, &1\leq j\leq N-1,\\
2i\sum_{p'\ne \frac{3\pi}{2}}(e^{ip'}-i)\ferm_{p'}\ferm_{3\pi/2}, \quad & j=N,
\end{cases}\label{eq:comm-TL-f}\\
[e_j,\e] &=
\begin{cases}
0, &1\leq j\leq N-1,\\
2i\sum_{p'\ne \frac{\pi}{2}}(e^{-ip'}-i)\fermd_{p'}\fermd_{\pi/2}, \quad & j=N.
\end{cases}\label{eq:comm-TL-e}
\end{align}
We thus see that the renormalized powers $\e$ and $\f$ are not
contained in the centralizer $\centJTL$ unless $L=1$ because of the last
Temperley--Lieb generator $e_N$ making the system
periodic. We note also that  for a finite chain the only powers of $\e$
  and $\f$ that commute with $\JTL{N}$ are $\e^{N/2}$ and
  $\f^{N/2}$. They are the highest non-zero powers and just mix the
  two  $\JTL{N}$-invariants -- the states with the all spins up or
  down.

To build elements in $\centJTL$, we can then modify the $\e$ and $\f$ by elements
 from the respective annulators of the
 commutators~\eqref{eq:comm-TL-f} and~\eqref{eq:comm-TL-e}. The first
 obvious candidates for the modifying elements are
 $\E\sim\fermd_{\pi/2}$ and $\F\sim\ferm_{3\pi/2}$
 (see~\eqref{QG-ferm-1}), respectively:
\begin{equation*}
[e_j,\f\,\F]=[e_j,\f]\F=0, \qquad [e_j,\e\E]=[e_j,\e]\E=0, \qquad 1\leq
j\leq N.
\end{equation*}
Moreover, there are  many other elements in $\LQG$ commuting with
all $e_j$'s:
\begin{equation}\label{eq:comm-TL-efn}
[e_j,\f^n \F]=[e_j,\f]\F \f^{n-1}=0, \qquad [e_j,\e^m \E]=[e_j,\e]\E \e^{m-1}=0, \qquad 1\leq j\leq N, 
\quad n,m\geq 1,
\end{equation}
which can be easily proved by induction.

In particular, we have the equality in $\LQG$,
\begin{equation*}
\f^n \F \e^m \E = \f^n \e^m \F\E,
\end{equation*}
which follows from the relations
\begin{equation*}
[\F,\e^m] = m\ffrac{\K+\K^{-1}}{2}\e^{m-1}\E, \qquad [\E,\f^n] = n\ffrac{\K+\K^{-1}}{2}\f^{n-1}\F, \quad\text{and}\quad \E^2=\F^2=0,
\end{equation*}
where the first two are obtained using~\eqref{Ef-rel}.

\begin{dfn}\label{dfn:Uqodd}
We now introduce the associative algebra $\LQGodd$, generated 
 $\FF n$, $\EE m$ ($n,m\in\oN\cup\{0\}$), $\K^{\pm1}$, 
$\h$ with the following defining relations
\begin{gather}
\K \EE m \K^{-1} = \q^2 \EE m,\qquad \K \FF n \K^{-1} = \q^{-2} \FF n,
\qquad \K^4=\one,\\
[\EE m,\FF n] = \sum_{r=1}^{\text{min}(n,m)}P_r(\h)\FF{n-r}\EE{m-r},\\
\EE{m}\EE{n} = \EE{n}\EE{m} = 0,\quad \FF{m}\FF{n}= \FF{n}\FF{m} =
0,\quad [\K,\h]=0,\label{Uqodd-dfn-3}\\
[\h,\EE m] = (m + \ffrac{1}{2})\EE{m}, \qquad [\h,\FF n] = -(n + \ffrac{1}{2})\FF{n},\label{Uqodd-dfn-4}
\end{gather}
where $P_r(\h)$ are polynomials on $\h$ from the usual $s\ell(2)$
relation
$[\e^m,\f^{n}]=\sum_{r=1}^{\text{min}(n,m)}P_r(\h)\f^{n-r}\e^{m-r}$,
and we assume that $\sum_{r=1}^{0}f(r)=0$.
\end{dfn}

The algebra $\LQGodd$ has the PBW basis $\EE n \FF m \h^k\K^l$, with
$n,m,k\geq0$ and $0 \leq l \leq 3$.
The positive Borel
subalgebra is generated by $\h$, $\K$ and $\EE n$ while the negative subalgebra -- by $\h$, $\K$ and $\FF n$, for $n\geq0$.

\begin{rem}\label{rem:inj-hom-LQGodd}
We note there is an injective homomorphism $\LQGodd\to\LQG$:
\begin{equation*}
\EE m \mapsto \e^m\E\,\ffrac{\K^2+\one}{2}, \qquad \FF n \mapsto \f^n\F\,\ffrac{\K^2+\one}{2}.
\end{equation*}
This  subalgebra in $\LQG$ can be realized as the limit $\q\to i$ of the renormalized \textit{odd}-powers of the $\E$ and $\F$ in $U_{\q}s\ell(2)$ at generic $\q$:
\begin{equation*}
\ffrac{\E^{2m+1}}{[2m+1]!} \xrightarrow[\;\q\to i\;]{} \e^{m}\E, \qquad
\ffrac{\F^{2n+1}}{[2n+1]!} \xrightarrow[\;\q\to i\;]{} \f^{n}\F, 
\qquad n,m\geq 0,
\end{equation*}
up to some irrelevant coefficients.
\end{rem}

We are now ready to formulate the main result of this section about the
centralizer of the image of $\JTL{2L}(0)$ under the representation $\repgl$.
\begin{thm}\label{Thm:centr-JTL-main}
On the alternating
periodic $\gl(1|1)$ spin chain $\Hilb_{2L}$,
the centralizer $\centJTL$ of the image of
Jones--Temperley--Lieb algebra $\repgl\bigl(\JTL{2L}(0)\bigr)$ (where
$\repgl$ is defined in~\eqref{rep-JTL-1} and~\eqref{rep-JTL-2}) is the  
 subalgebra in $\repQG\bigl(\LQG\bigr)$
generated by $\LQGodd$ and $\f^L$, $\e^L$. 
\end{thm}

The full proof of this statement is too long and has been relegated to
App.~B.

\subsubsection{Fermion expression for the centralizer $\centJTL$}\label{sec:ferm-exp-cent}
We note here that generators of $\centJTL$  in  Thm.~\ref{Thm:centr-JTL-main}
have a simple fermionic expression, for $n\geq0$,
\begin{eqnarray}\label{Fodd-sym-def}
F_{(2n+1)}&=&\!\!\!\!\sum_{\substack{1\leq
  j_1<j_2<\,\dots\\\dots\,<j_{2n+1}\leq 2L}}f_{j_1}f_{j_2}\dots f_{j_{2n+1}},\\
F^\dagger_{(2n+1)}&=&\!\!\!\!\sum_{\substack{1\leq
  j_1<j_2<\,\dots\\\dots\,<j_{2n+1}\leq
  2L}}f^\dagger_{j_1}f^\dagger_{j_2}\dots f^\dagger_{j_{2n+1}},\\
F_{(2L)}&=&f_{1}f_{2}\dots f_{2L},\nonumber\\
F^\dagger_{(2L)}&=&f^\dagger_{1}f^\dagger_{2}\dots f^\dagger_{2L},\nonumber\\
{\N}&=&\sum_{1\leq j\leq 2L} (-1)^jf_j^\dagger f_j-L\nonumber,
\end{eqnarray}
which is to be compared with the generators~\eqref{FF-sym-def} of the centralizer ${\cal
  A}_{1|1}$ in the open case. The correspondence
  with the generators of $\LQGodd$ is $F_{(2n+1)} = \frac{\q^{-n+1}}{n!}\repQG(\FF{n})$,
  $F^\dagger_{(2n+1)} = \frac{\q^n}{n!}\repQG(\EE{n}\K^{-1})$, with $n>0$, while
  $n=0$ correspondence is given in~\eqref{QG-fermf-1}, and $\N$ is
  proportional to $\repQG(\h)=S^z/2$.

\medskip

In our second paper~\cite{GRS2}, we  rely on representation theory of the
$\JTL{N}$-centralizer $\centJTL$ in order to study the  decomposition of the
periodic spin-chain into indecomposable $\JTL{N}$-modules. 

\subsection{A note on the twisted model}\label{sec:twist-mod}
We can also  consider the antiperiodic model
for the $\gl(1|1)$ chain, obtained by setting
$f^{(\dagger)}_{2L+1}=-f^{(\dagger)}_1$. The generators $e_j$, for $1\leq j\leq 2L-1$, have the same representation~\eqref{rep-JTL-1} while the last generator is then given by
\begin{equation*}
e_{2L}=(f_{2L}-f_{1})(f_{2L}^\dagger-f_{1}^\dagger),
\end{equation*}
to be compared with~\eqref{rep-JTL-2}.
 This does not provide more  a
representation of   the $\JTL{N}$
algebra but rather a representation of an abstract algebra generated by
 $e_j$   and $u^{\pm2}$  
 with the  relations~\eqref{TL} for
 $1\leq j\leq N$
  and~\eqref{aTL-rel}, among others. 
 We will call the
corresponding algebra $\JTL{N}^{tw}$. The corresponding XX spin chain
now is periodic for even spin, and antiperiodic for odd spin.
Note that the action of  $\JTL{N}^{tw}$ \textit{does not} commute with $\gl(1|1)$ generators $F_{(1)}$ and $F_{(1)}^{\dagger}$  (or $\F$ and $\E$, equivalently) defined in~\eqref{FF-sym-def}. Therefore, the  hamiltonian densities algebra does not have $\gl(1|1)$ symmetry in this case.

We next study the centralizer of the representation of $\JTL{N}^{tw}$.
It turns out that the choice of ``even" subalgebra in $U_{\q}s\ell(2)$ at generic $\q$,
{\it i.e.}, the algebra generated by the renormalized \textit{even}-powers of the $\E$
and $\F$ gives in the limit $\q\to i$ the centralizer for the representation of $\JTL{N}^{tw}$
on the spin-chain with the opposite twist --- the usual $U(s\ell(2))$
generated by the $\e$ and $\f$. The  proof is given  below.

\begin{thm}\label{Thm:centr-JTL-anti}
On the alternating antiperiodic $\gl(1|1)$ spin chain,
 the centralizer of the image
of the representation of the algebra $\JTL{N}^{tw}$ is the associative algebra $\repQG(U s\ell(2))$.
\end{thm}
\begin{proof}
We first check using expressions~\eqref{QG-fermf-2} for the $U s\ell(2)$ generators in terms of $f_j$ and $f^{\dagger}_j$ fermions that the  action of $U s\ell(2)$ indeed commutes with the additional generator $e_N=(f_N-f_1)(f^{\dagger}_N-f^{\dagger}_1)$; that the generators $e_j$, for $1\leq j\leq N-1$, commute with the $U s\ell(2)$ is obvious because the centralizer of the  $\TL{N}$ contains $\repQG(\e)$ and $\repQG(\f)$. Next, a simple calculation using again the $f_j$ and $f^{\dagger}_j$ fermions shows that the $e_N$ does not commute with the operators $\repQG(\e^n\f^m\h^k\F)$, $\repQG(\e^n\f^m\h^k\E)$, $\repQG(\e^n\f^m\h^k\F\E)$, for $n,m,k\geq0$. To show that there are no linear combinations of these operators in the centralizer, we go to the Fourier transforms as in Sec.~\ref{sec:fourier} introducing $\ferm_p$ and $\fermd_p$ with the same formal expression~\eqref{def-ferm} but now the momenta $p_m$ takes values $\frac{2\pi m}{N}$ for $L$ odd and $\frac{(2m-1)\pi}{N}$ for $L$ even. We then carry out  calculations fully similar to those  in the proof of Thm.~\ref{Thm:centr-JTL-main} (which are mainly presented in Lem.~\ref{lem:endo-perTL-count}). Additional care should be taken
 in handling fermionic expressions for the $\LQG$ generators in terms of $\ferm_p$ and $\fermd_p$, which are different from the ones in~\eqref{QG-ferm-1}-\eqref{QG-ferm-3}. One proves easily in this way that the centralizer of the algebra generated by $e_j$, for $1\leq j\leq N$, in the antiperiodic spin-chain is given by $\repQG(U s\ell(2))$.
 
 Finally, we show that the generators $u^{\pm2}$ commute with the action of $U s\ell(2)$. The $u^2$ acts on the fermions $f_j$ and $f_j^{\dagger}$ formally in the same way~\eqref{gen-u-acts} as in the periodic model but it changes sign in front of $f^{(\dagger)}_{j+2-N}$ whenever the position $j+2$ is greater than $N$ due to the antiperiodic conditions. We then obtain 
 \begin{equation*}
 u^2\repQG(\f)u^{-2} = \q \sum_{1\leq j_1<j_2\leq N} f_{j_1+2}f_{j_2+2} =  \q \sum_{1\leq j_1<j_2\leq N} f_{j_1}f_{j_2} = \repQG(\f)
 \end{equation*}
 and similarly for $\repQG(\e)$. This finishes the proof.
\end{proof}

We emphasize that the antiperiodic $\gl(1|1)$ spin chain does not
have $\gl(1|1)$ symmetry any longer. We will come back briefly to this
twisted case in other subsections -- the main text meanwhile is only
devoted to the periodic case.

\section{The scaling limit of the closed $\gl(1|1)$ chains}\label{sec:scaling}


In this Section, we discuss how to proceed from the $\JTL{N}$
generators to get Virasoro modes in the non-chiral logarithmic
conformal field theory of symplectic fermions:
we show that the combinations 
\begin{equation}\label{HPn-def}
H(n) = -\sum_{j=1}^N e^{-iqj} e^{\gl}_j,\qquad
P(n)=\ffrac{i}{2}\sum_{j=1}^N e^{-iqj} [e^{\gl}_j,e^{\gl}_{j+1}], \qquad q=\ffrac{n\pi}{L},
\end{equation}
of the (representation of)
$\JTL{N}$ generators converge in a certain sense (the scaling limit) as $L\to\infty$ to the well-known symplectic fermions representation
of the left and right Virasoro generators
\begin{equation*}
\ffrac{L}{2\pi}H(n) \mapsto 
L_{n}+\bar{L}_{-n}, \qquad \ffrac{L}{2\pi}P(n) \mapsto 
L_{n}-\bar{L}_{-n}.
\end{equation*}

 For convenience, we
begin with studying the $\gl(1|1)$-Hamiltonian spectrum on a finite
lattice in Sec.~\ref{sec:ham-spec} and Sec.~\ref{sec:ham-spec-Jordan}, where we also introduce 
technically  more suitable lattice fermions. 
We  then give a formal definition of the scaling limit procedure in Sec.~\ref{sec:emerge-latVir} and show the convergence of
the whole family of lattice higher Hamiltonians (with their Fourier
transformations) to all generators of the product $\VirN(2)=\Vir(2)\boxtimes\overline{\Vir}(2)$ of
 the left and right Virasoro algebras with the central charge $c=-2$. The
important result that the scaling limit
respects algebraic relations is discussed in Sec.~\ref{sec:fromJTL-Vir}.

\subsection{The Hamiltonian and  $\chi$-$\eta$ fermions}\label{sec:ham-spec}

We now go back to  the periodic $\gl(1|1)$ spin-chain with the following
$\JTL{N}$-representation:
 \begin{equation*}
 e^{\gl}_j=\left(f_j+f_{j+1}\right)\left(f_j^\dagger+f_{j+1}^\dagger\right),
 \qquad f_{N+1}=f_1,\quad  f^{\dagger}_{N+1}=f^{\dagger}_1,\qquad 1\leq j\leq N,
 \end{equation*}
which is discussed above in Sec.~\ref{sec:super-spin-ch-def} and Sec.~\ref{sec:fourier}.  We
abuse the notation for the representation of the $\JTL{N}$ generators in
what follows and write simply $e_j$ instead of $e^{\gl}_j$.
Setting 
\begin{equation*}
f_j^\dagger=i^{j}c_j^\dagger,\qquad
f_j=i^{j}c_j
\end{equation*}
we get as well
\begin{equation}
e_j=i(-1)^j\left[c_jc^\dagger_{j+1}-c_{j}^\dagger c_{j+1}
  +i(c_j^\dagger c_j-c_{j+1}^\dagger c_{j+1})\right] ,\qquad 1\leq
  j\leq 2L.
\end{equation}

We find it more convenient to use Fourier transforms of the fermions
$c_j$ and $c^\dagger_j$, and set
\begin{equation}
c_j=\ffrac{1}{\sqrt{N}} \sum_{p_m} e^{ijp_m}\theta_{p_m},\qquad
c^\dagger _j=\ffrac{1}{\sqrt{N}} \sum_{p_m} e^{-ijp_m}\theta^\dagger_{p_m},
\end{equation}
where the sums are taken over all the momenta $p_m$ introduced in~\eqref{momenta-set}. We obtain 
 then the Hamiltonian
\begin{equation}\label{eq:hamilt-fermd}
H = -\sum_{j=1}^{2L}e_j = 2\sum_{p}(1+\sin{p})\fermd_p\,\ferm_{\pi+p},
\end{equation}
which can be rewritten in (almost) diagonal form:
\begin{equation}\label{eq:def-hamilt-gl}
H=
2\sum_{\substack{p=\step\\\text{step}=\step}}^{\pi-\step}\sin{p}\bigl(\chi^{\dagger}_p\chi_{p}
- \eta^{\dagger}_p\eta_{p}\bigr) + 4 \chi^{\dagger}_0\eta_0,
\end{equation}
where $\step=\frac{2\pi}{N}$ and we introduced
\begin{align}
\chi^{\dagger}_p &=
\ffrac{1}{\sqrt{2}}\Bigl(\sqrt{\tan{\ffrac{p}{2}}}\,\fermd_{p-\frac{\pi}{2}} +
\sqrt{\cot{\ffrac{p}{2}}}\,\fermd_{p+\frac{\pi}{2}}\Bigr),&
\chi_p &=
\ffrac{1}{\sqrt{2}}\Bigl(\sqrt{\cot{\ffrac{p}{2}}}\,\ferm_{p-\frac{\pi}{2}} +
\sqrt{\tan{\ffrac{p}{2}}}\,\ferm_{p+\frac{\pi}{2}}\Bigr),\notag\\
\eta^{\dagger}_p &=
\ffrac{1}{\sqrt{2}}\Bigl(\sqrt{\tan{\ffrac{p}{2}}}\,\fermd_{p-\frac{\pi}{2}} -
\sqrt{\cot{\ffrac{p}{2}}}\,\fermd_{p+\frac{\pi}{2}}\Bigr),&
\eta_p &=
\ffrac{1}{\sqrt{2}}\Bigl(\sqrt{\cot{\ffrac{p}{2}}}\,\ferm_{p-\frac{\pi}{2}} -
\sqrt{\tan{\ffrac{p}{2}}}\,\ferm_{p+\frac{\pi}{2}}\Bigr),\label{eq:chi-eta-def}\\
\chi^{\dagger}_0 &= \fermd_{\frac{\pi}{2}},\quad \chi_0 = \ferm_{\frac{\pi}{2}},&
\eta^{\dagger}_0 &= \fermd_{\frac{3\pi}{2}},\quad \eta_0=\ferm_{\frac{3\pi}{2}},\notag
\end{align}
with momenta $p$ shifted by $\pi/2$ and taking thus values $p=p_n=\step n$, where $1\leq n\leq L-1$, for even and odd $L$. The normalizations have been chosen to ensure relativistic
dispersion relation with unit speed of light, and  to satisfy the anti-commutation relations
\begin{equation*}
\left\{\chi^{\dagger}_p,\chi_{p'}\right\} = \left\{\eta^{\dagger}_p,\eta_{p'}\right\} =
\delta_{p,p'}, \qquad 
\left\{\chi_p,\eta_{p'}\right\} =
\left\{\chi^{\dagger}_p,\eta^{(\dagger)}_{p'}\right\} =
\left\{\eta^{\dagger}_p,\chi^{(\dagger)}_{p'}\right\} = 0. 
\end{equation*}
For  convenience, we also give expressions for  $\ferm^{(\dagger)}$s in terms of $\chi^{(\dagger)}$s and
$\eta^{(\dagger)}$s,
\begin{equation}\label{eq:ferm-chi-eta}
\begin{split}
\fermd_{p'-\frac{\pi}{2}} &=
\sqrt{\ffrac{\cot{(p'/2)}}{2}}\bigl(\chi^{\dagger}_{p'} +
\eta^{\dagger}_{p'}\bigr), \qquad
\ferm_{p'-\frac{\pi}{2}} = \sqrt{\ffrac{\tan{(p'/2)}}{2}}
\bigl(\chi_{p'} + \eta_{p'}\bigr),\\
\fermd_{p'+\frac{\pi}{2}} &=
\sqrt{\ffrac{\tan{(p'/2)}}{2}}\bigl(\chi^{\dagger}_{p'}
- \eta^{\dagger}_{p'}\bigr),\qquad
\ferm_{p'+\frac{\pi}{2}} = \sqrt{\ffrac{\cot{(p'/2)}}{2}}
\bigl(\chi_{p'} - \eta_{p'}\Bigr),
\end{split}\qquad \step\leq p'\leq\pi-\step.
\end{equation}

\subsection{Hamiltonian spectrum and Jordan blocks}\label{sec:ham-spec-Jordan}
We now study the  spectrum of the Hamiltonian~\eqref{eq:def-hamilt-gl}
and analyze the  Jordan blocks appearing  on a finite lattice. Once the Hamiltonian is written as a quadratic form in free fermionic modes as in~\eqref{eq:def-hamilt-gl}, 
the zero-mode term $\chi^{\dagger}_0\eta_0$ (which  is proportional to
the Casimir operator of the quantum group $U_{\q}s\ell(2)$) implies
the existence of non-trivial Jordan blocks since, for a given set of filled modes at non zero momentum, the action of the operators $\chi_0^{(\dagger)}$ and $\eta_0^{(\dagger)}$ allows one to build a four dimensional subspace with the same energy, and Jordan block of dimension two analogous to the one for the Casimir.

We first note that the diagonal part $H^{(d)}$ of the Hamiltonian has
the eigenvectors
\begin{equation}
v(\{p_k\},\{p'_j\}) =
\prod_{{\{p_k\}}}\eta_{p_k}\prod_{{\{p'_j\}}}\chi_{p'_j}\vectv{\uparrow\dots\uparrow},
\end{equation}
where $\vectv{\uparrow\dots\uparrow}$ is the state with all spins up, with the eigenvalues
\begin{equation}\label{eigenval-H}
2\sum_{p\in\{p_k\}}\sin{p}-2\sum_{p\in\{p'_j\}}\sin{p},
\end{equation}
where the sets $\{p_k\}$ and $\{p'_j\}$ are any subsets in
the set $\{p_n=\pi n/L,\;1\leq n\leq L-1\}$ of allowed momenta. We thus immediately
find the four ground states
\begin{equation}
\phip =
\prod_{\substack{p=\step\\\text{step}=\step}}^{\pi-\step}\chi_{p}\vectv{\uparrow\dots\uparrow},
\qquad \phim = \chi_0\eta_0\phip,
\qquad \vac = \eta_0\phip,
\qquad \lvac = \chi_0\phip,
\end{equation}
where the two fermionic states $\phip$ and $\phim$ belong to the
sectors with $S^z=+1$ and $S^z=-1$, respectively, and the two bosonic
states $\vac$ and $\lvac$ have $S^z=0$.

What is crucial for logarithmic CFT is to know the structure of Jordan
blocks. The Hamiltonian we study has the off-diagonal part
$\chi^{\dagger}_0\eta_0$ which generates Jordan blocks of rank $2$. For
example, the space of ground states has the following structure:
\begin{equation}\label{diag:space-gr-st}
   \xymatrix@=20pt{
     &\lvac\ar[dl]_{\chi^{\dagger}_0} \ar[dr]^{-\eta_0} \ar[dd]^(0.45){H}&&\\
     \phip\ar[dr]_{\eta_0}
       &&\phim\ar[dl]^{\chi^{\dagger}_0} &\\
     &\vac&&
   }   
\end{equation}
where the vacuum $\vac$ and the state $\lvac$ form a
two-dimensional Jordan cell of the lowest eigenvalue for~$H$. We also
show the action of $\F\sim\eta_0$
and $\E\sim\chi^{\dagger}_0$ in~\eqref{diag:space-gr-st}.

The whole space of states $\Hilb_{2L}$ is generated from one cyclic
vector $\lvac$ by the algebra of creation modes (including the zero
modes generating the vacuum subspace)
\begin{equation}\label{Falg-def}
\Falg = \bigl\{\,\chi^{\dagger}_p,\; \eta_p\,;\;\; p=\pi n/L,\; 0\leq n<L\,\bigr\}.
\end{equation}
The annihilation modes are
\begin{equation}
\chi_{p}\vac = \eta^{\dagger}_p\vac = \chi^{\dagger}_0\vac =
\eta_0\vac =0, \qquad p\in \{\pi n/L,\;1\leq n\leq L-1\}.
\end{equation}

\subsection{Emergence of the left and right Virasoro algebras}\label{sec:emerge-latVir}
In this section, we study  the scaling limit properties of
 the periodic spin-chain in detail. Recall that an essential ingredient  in the general definition
 of the scaling limit sketched in Sec.~\ref{prelim:scal-lim} is 
 the low-lying eigenstates of the Hamiltonian $H$. In order to study the 
 action of JTL elements on these eigenstates in the limit $L\to\infty$
 (recall $N=2L$) we first truncate each $\Hilb_{2L}$, keeping only
 eigenspaces up to an energy level $M$, for each positive number $M$.
 Each such truncated space turns out to be finite-dimensional in the
 limit, {\it i.e.}, it depends on $M$ but not $L$. Then, keeping matrix
 elements of JTL elements that correspond to the action only within
 these truncated spaces of scaling states, we obtain well-defined
 operators in the limit $L\to\infty$. The corresponding operators
 acting on all scaling states of the CFT can be finally obtained (if
 they exist) in the second limit $M\to\infty$.

To put things a little  more formally, we define \textit{the scaling limit} denoted
simply by `$\mapsto$' as a limit over graded spaces of coinvariants with
respect to smaller and smaller subalgebras in the creation modes
algebra~$\Falg$ introduced in~\eqref{Falg-def}, {\it i.e.}, along the
following lines:
\begin{enumerate}
\item[\textit{1.}] we  consider a family of subalgebras
  $\Falg[M]\subset\Falg$ generated by the creation modes
  $\chi^{\dagger}_p$ and $\eta_p$ in the range $M\step<p<\pi-M\step$,
  where $0\leq M \leq L'/2$ and we set $L'=L-(L\modd2)$ and recall
  $\step={\pi/ L}$; we thus have a tower of subalgebras
  \begin{equation}\label{Falg-tower}
  0=\Falg[L'/2]\subset\Falg[L'/2-1]\subset\dots\subset\Falg[2]\subset\Falg[1]\subset\Falg[0]\subset\Falg;
  \end{equation}
  
\item[\textit{2.}] we consider then vector-spaces
 $\Hilb_{2L}/\Falg[M]\Hilb_{2L}$ of coinvariants\footnote{Here,
 $\Falg[M]\Hilb_{N}$ means the image of the action of the whole
 algebra $\Falg[M]$ on $\Hilb_N$. Then, coinvariants by definition are
 elements of the quotient-space $\Hilb_N/\Falg[M]\Hilb_N$.} graded by
 the Hamiltonian $H$, for any finite $M$. Note also that for each
 fixed $M$ these graded spaces are stabilized after some $L=L_0$ and
 they are \textit{finite-dimensional} at $L\to\infty$; each of these stabilized
 spaces we denote as $\spco_M$.  In physical terms, we keep only the
 low energy modes, which are those close to $0$ and $\pi$.
  
\item[\textit{3.}] we compute Fourier transforms of $e_j$'s and
$[e_j,e_{j+1}]$'s corresponding to finite modes on the
finite-dimensional graded vector-spaces of coinvariants $\spco_M$ in
the limit $L\to\infty$ (physically, we keep only long wave-length
contribution to low-lying excitations over the ground states). By  
computating in the limit $L\to\infty$ we mean here showing  strong
convergence\footnote{The strong convergence of operators requires a
  normed vector space, or positive-definite inner product. One can
  introduce this inner product here using the fact that a finite-dimensional
  vector space with \textit{non-degenerate} indefinite inner product is a
  Krein space and, therefore,  can be turned into a
  positive-definite inner product space~\cite{Bognar}. This applies to  $\Hilb_{2L}$ endowed with non-degenerate indefinite inner product $\inn{\cdot}{\cdot}$ such that $\inn{f_j x}{y}=\inn{x}{f_j^{\dagger}y}$ for any $x,y\in\Hilb_{2L}$.} of the sequence of operators (the Fourier transforms)
  parametrized by $L$ towards  a particular operator
acting on $\spco_M$.

\item[\textit{4.}] we finally  take a limit with respect to smaller and
smaller subalgebras $\Falg[M]$ in the tower~\eqref{Falg-tower}, {\it i.e.}, we take the second limit
$M\to\infty$. 
So for the spaces of low-lying states $\{\spco_M, \, M\geq1\}$, we take an inductive\footnote{It is more natural to take a projective limit for the spaces of coinvariants, as they are defined as quotients, but this limit is then a completion of the vector space $\spco_{\infty}$ of physical states.} limit which gives the space $\spco_{\infty}$ of all scaling states. This is an infinite-dimensional Krein space, {\it c.f.}~\cite{Iv}, which has a positive-definite inner product. In this space one can then study convergence of operators in the second limit\footnote{One could then go back to the original indefinite inner product, which is used in LCFT, using the Krein space structure on $\spco_{\infty}$ or the so-called fundamental symmetry of the Krein space~\cite{Bognar}.}.
\end{enumerate}
Note that we could equivalently consider the same construction/definition of the scaling limit based on a slightly different tower of  subalgebras $\tilde{\Falg}[M]\subset\Falg$ which generate all eigenstates  between the energy level $M+1$ and the maximum one. But then a definition of each $\tilde{\Falg}[M]$ is more complicated: it is generated by all monomials $\prod_{{\{p_k\}}}\eta_{p_k}\prod_{{\{p'_j\}}}\chi^{\dagger}_{p'_j}$ such that
$2\sum_{p\in\{p_k\}}\sin{p}+2\sum_{p\in\{p'_j\}}\sin{p}> 2\sin{M\step}$ (recall the eigenvalues in~\eqref{eigenval-H}). This choice is probably more natural, in view of the discussion in the beginning of this subsection, but the first choice~\eqref{Falg-tower} of the tower of the subalgebras $\Falg[M]$, which is  much simpler technically,   is enough  for  the purposes of this paper.

\subsubsection{The scaling limit of the Hamiltonian}
Following the lines \textit{1.-4.} in the definition above,
we first study the scaling limit of the
Hamiltonian~\eqref{eq:def-hamilt-gl}. We rewrite it in the
normal-ordered form as
\begin{equation}\label{eq:ham-chi-eta-norm}
H=
2\sum_{\substack{p=\step\\\text{step}=\step}}^{\pi-\step}\sin{p}\bigl(\chi^{\dagger}_p\chi_{p}
+ \eta_p\eta^{\dagger}_{p}\bigr) + 4 \chi^{\dagger}_0\eta_0 - 2\sum_{p=\step}^{\pi-\step}\sin{p},
\end{equation}
where we explicitly extracted the ground-state energy in the last sum.
We can now {\sl linearize} the dispersion relation around $p=0$ and
$p=\pi$ in the first limit $N\to\infty$ introducing the left-moving modes
$\bar{\chi}_p^{(\dagger)}=\chi_{\pi-p}^{(\dagger)}$ and
$\bar{\eta}_p^{(\dagger)}=\eta_{\pi-p}^{(\dagger)}$. 
The excitations over the Dirac sea are thus described by
\begin{equation}
H= \ffrac{4\pi}{N}\sum_{m>0}m\left(\chi_p^\dagger\chi_p + \bar{\chi}_p^\dagger\bar{\chi}_p\! + \eta_p\eta^\dagger_p
+ \bar{\eta}_p\bar{\eta}^\dagger_p\right)+4\chi^{\dagger}_0\eta_0 +
\vacl H\vacr,\qquad p\equiv\ffrac{m\pi}{L}+o\,(1/N).
\end{equation}
 The ground-state
energy has the following leading asymptotics in the large-$N$ limit,
\begin{equation*}
\vacl H\vacr = -2\sum_{m=1}^{L-1}\sin{\ffrac{m\pi}{L}} = -2\cot{\ffrac{\pi}{N}} =
-\ffrac{2N}{\pi} + \ffrac{2\pi}{3N} + o\,(1/N),
\end{equation*}
where we used the trigonometric identity
\begin{equation*}
\sum_{k=1}^n \sin{k\alpha} = \ffrac{\sin{\frac{(n+1)\alpha}{2}}\sin{\frac{n\alpha}{2}}}{\sin{\frac{\alpha}{2}}}.
\end{equation*}
We thus obtain the expansion
\begin{equation}
H= H^{(d)} + H^{(n)} =-\ffrac{2N}{\pi} +
\ffrac{4\pi}{N}\left((L_0+\bar{L}_0)^{(d)} + (L_0+\bar{L}_0)^{(n)} -\ffrac{c}{12}\right) +
 o\,(1/N),
\end{equation}
with the central charge $c=-2$. The diagonal part of the Hamiltonian
in the scaling limit is
\begin{equation*}
(L_0+\bar{L}_0)^{(d)} = \sum_{m>0}m\left(\chi_p^\dagger\chi_p + \bar{\chi}_p^\dagger\bar{\chi}_p\! + \eta_p\eta^\dagger_p
+ \bar{\eta}_p\bar{\eta}^\dagger_p\right),\qquad p\equiv\ffrac{m\pi}{L},
\end{equation*}
and the non-diagonal part is
\begin{equation*}
(L_0+\bar{L}_0)^{(n)} = \ffrac{N}{\pi}\chi^{\dagger}_0\eta_0.
\end{equation*}
Finally, we introduce some other notation convenient for the scaling
limi \footnote{The distinction between $L$
even and odd disappears as the new moments are defined with respect to
$\frac{\pi}{2}$.},
\begin{eqnarray}\label{eq:sferm-lat-def}
\begin{split}
\sfermm_m=\sqrt{m}\,\chi_p,\quad\!\sfermp_m=\sqrt{m}
\,\bar{\eta}^\dagger_p,\quad\!\bsfermm_m=\sqrt{m}\,\bar{\chi}_p,\quad\!\bsfermp_m=-\sqrt{m}
\,\eta^\dagger_p,\quad\! \sfermp_0=\bsfermp_0=\sqrt{\ffrac{L}{\pi}}\,\chi_0^{\dagger}=\sqrt{\ffrac{L}{\pi}}\,\theta^\dagger_{{\pi\over 2}},\quad\mbox{}\\
\sfermm_{-m}=-\sqrt{m}\,\bar{\eta}_p, \quad\!\sfermp_{-m}=\sqrt{m}\,\chi^\dagger_p,
\quad\!\bsfermm_{-m}=\sqrt{m}\,\eta_p, \quad\!\bsfermp_{-m}=\sqrt{m}\,\bar{\chi}^\dagger_p,
\quad\! \sfermm_0=\bsfermm_0=\sqrt{\ffrac{L}{\pi}}\,\eta_0=\sqrt{\ffrac{L}{\pi}}\,\theta_{{3\pi\over 2}},
\end{split}
\end{eqnarray}
t, for $m>0$ and $p=m\pi/L$. One has now the anti-commutation relations
\begin{equation*}
\{\psi^{\alpha}_m,\psi^{\beta}_{m'}\} =
mJ^{\alpha\beta}\delta_{m+m',0}\,,\qquad \alpha,\beta\in\{1,2\},
\end{equation*}
with the symplectic form $J^{12}=-J^{21}=1$.
So, we get the scaling limit
\begin{multline}
\ffrac{L}{2\pi}\Bigl(H+\ffrac{2N}{\pi}\Bigr)\mapsto L_0+\bar{L}_0-\ffrac{c}{12}= \sum_{m>0}
\left(\sfermp_{-m}\sfermm_m-\sfermm_{-m}\sfermp_m + \,\psi\to\bar{\psi}\right)
+ 2\sfermp_0\sfermm_0 -\ffrac{c}{12}\\= \sum_{m\in\oZ} :\!\sfermp_{-m}\sfermm_m\!:+ \bigl(\psi\to\bar{\psi}\bigr)-\ffrac{c}{12}.
\end{multline}
This expression of $L_0+\bar{L}_0$ is well known and appears in the
theory of symplectic fermions~\cite{Kausch}
\begin{eqnarray}\label{L0-def}
L_0=\sfermp_0\sfermm_0 + \sum_{m>0}
\left(\sfermp_{-m}\sfermm_m-\sfermm_{-m}\sfermp_m\right)=
\sum_{m\in\oZ} :\!\sfermp_{-m}\sfermm_m\!: \,,\nonumber\\
\bar{L}_0=\sfermp_0\sfermm_0 + \sum_{m>0}
\left(\bsfermp_{-m}\bsfermm_m - \bsfermm_{-m}\bsfermp_m\right) =
\sum_{m\in\oZ} :\!\bsfermp_{-m}\bsfermm_m\!: \,,\nonumber
\end{eqnarray}
where $L_0$ and $\bar{L}_0$ have a common part, made of
$\sfermp_0\sfermm_0$\,.

\subsubsection{The momentum operator}
We next obtain the conformal spin operator $L_0-\bar{L}_0$ using
lattice calculations. The general mapping~\cite{KooSaleur} between anisotropic transfer
matrices and evolution operators in CFT suggests that a lattice analogue
of $T_{xy}$, the off-diagonal component of the stress tensor,
is a momentum operator
\begin{equation}\label{lattmom}
P=\ffrac{i}{2}\sum_{j=1}^{N} [e_j,e_{j+1}].
\end{equation}

Straightforward calculations for the $\gl(1|1)$ spin chain give
\begin{equation}
[e_j,e_{j+1}]=-i\left[c_jc_{j+1}^\dagger-c_{j+1}c_j^\dagger-c_{j+1}c_{j+2}^\dagger+c_{j+2}c_{j+1}^\dagger+i(c_jc_{j+2}^\dagger -c_{j+2}c_j^\dagger)\right],
\end{equation}
so the momentum reads, in terms of fermion Fourier variables
\begin{equation}
P =
\sum_{\substack{p=\step\\\text{step}=\step}}^{\pi-\step}\sin{2p}\bigl(\chi^{\dagger}_p\chi_{p}
+ \eta^{\dagger}_p\eta_{p}\bigr).
\end{equation}
The scaling limit of the rescaled operator $\ffrac{L}{2\pi}P$ gives
the conformal spin operator $L_0-\bar{L}_0$, keeping only the leading
term:
\begin{equation*}
\ffrac{L}{2\pi} P \mapsto \sum_{m\in\oZ} :\!\sfermp_{-m}\sfermm_m\!:-
\bigl(\psi\to\bar{\psi}\bigr) = L_0-\bar{L}_0.
\end{equation*}

We also note  that the generator  $u^2$ of translations is simply
related with the momentum in the continuum limit.
Going to the $\eta$ and $\xi$ modes in~\eqref{uu-theta} and using repeatedly that $e^{2i\pi}=1$ to shift summation leads  to
\begin{equation}
u^2= \exp\left[-2i\sum_{p=\epsilon}^{\pi-\epsilon} p(\chi_p^\dagger\chi_p+\eta_p^\dagger\eta_p)\right],
\qquad \step=\ffrac{2\pi}{N},
\end{equation}
with the step $\step$ in the sum.
The term in the exponential is a linearized version of the momentum~$P$.

\subsubsection{Higher Virasoro modes}
It is interesting to obtain  expressions for all other modes $L_n$ and
$\bar{L}_n$ of the stress tensor by sticking to the
lattice some more. We consider the Fourier transform of $e_j$,
\begin{multline}
H(n) = -\sum_{j=1}^N e^{-iqj} e_j=\sum_{p}
\left[1+e^{iq}+ie^{-ip}-ie^{i(p+q)}\right] \fermd_{p}\,
\ferm_{p+q+\pi}\\
=4e^{iq/2}\Bigl(\sum_{\substack{p=\step\\\text{step}=\step}}^{\pi-\step}
\bigl(\sin{\ffrac{p+q}{2}}
\sin{\ffrac{p}{2}}\,\fermd_{p-\frac{\pi}{2}}\ferm_{p+q+\frac{\pi}{2}} +
[p\to p+\pi]\bigr) +
\cos{\ffrac{q}{2}}\,\fermd_{\frac{\pi}{2}}\ferm_{q-\frac{\pi}{2}}
\Bigr), \quad q={\ffrac{n\pi}{L}},\label{inter}
\end{multline} 
where $n$ is integer.
This sum can be split into
the two sums
$\sum_{\step}^{\pi-q-\step}$ and $\sum_{\pi-q+\step}^{\pi-\step}$
to be sure that the subscript $p'$ in the terms
$\ferm_{p'\pm\frac{\pi}{2}}$ takes values between $\step$ and
$\pi-\step$ which is necessary to use the
notations~\eqref{eq:ferm-chi-eta}. We first consider the case $0< n< L$.
Using the formulas~\eqref{eq:ferm-chi-eta} expressing the
$\theta^{(\dagger)}$s in terms of the $\chi^{(\dagger)}$s and
$\eta^{(\dagger)}$s, the $H(n)$ can be rewritten as
\begin{multline}\label{lat-Vir-H}
H(n) =
2e^{iq/2}\Biggl(\sqrt{\sin{q}}\, \chi^{\dagger}_{0}\bigl(\chi_{q} +
\eta_{q}\bigr) + \sum_{\substack{p=\step\\\text{step}=\step}}^{\pi-q-\step}
\sqrt{\sin{(p)}\sin{(p+q)}} \bigl(\chi^{\dagger}_{p}\,\chi_{p+q} -
\eta^{\dagger}_{p}\,\eta_{p+q}\bigr)\\ 
+ \sum_{\substack{p=\step\\\text{step}=\step}}^{q-\step}
\sqrt{\sin{(p)}\sin{(q-p)}}\bigl(\chi^{\dagger}_{\pi-p}\,\eta_{q-p} +
\eta^{\dagger}_{\pi-p}\,\chi_{q-p}\bigr)
+ \sqrt{\sin{q}}\, \bigl(\chi^{\dagger}_{\pi-q} +
\eta^{\dagger}_{\pi-q}\bigr)\eta_{0}\Biggr).
\end{multline}
Using the transformation~\eqref{eq:sferm-lat-def} to the fermions
$\psi^{1,2}$ and linearizing the dispersion relation, we thus have in
the scaling limit (keeping the low- and high-$p$ terms which have momenta close to $0$ or $\pi$, 
following the lines \textit{1.-4.} in the definition in Sec.~\ref{sec:emerge-latVir}), with a finite
mode $n$,
\begin{multline*}
\ffrac{L}{2\pi}H(n)\mapsto\sfermp_{0}\bigl( \sfermm_{n}+\bsfermm_{-n}\bigr)
+ \sum_{m>0}\bigl(\sfermp_{-m}\sfermm_{m+n} +
\sfermp_{m+n}\sfermm_{-m} + \bsfermp_{m}\bsfermm_{-m-n} +
\bsfermp_{-m-n}\bsfermm_{m}\bigr) \\
+ \sum_{m=1}^{n-1}\bigl(\sfermp_{m}\sfermm_{n-m} +
\bsfermp_{-m}\bsfermm_{m-n}\bigr)
+ \bigl(\sfermp_{n}+\bsfermp_{-n}\bigr)\sfermm_{0}.
\end{multline*}
We finally obtain the contribution corresponding to low-lying
excitations over the ground state,
\begin{equation}\label{scal-lim-Hn}
\ffrac{L}{2\pi}H(n) \mapsto \sum_{m\in\oZ}\sfermp_{-m}\sfermm_{m+n} +
\sum_{m\in\oZ}\bsfermp_{-m}\bsfermm_{m-n} = 
L_{n}+\bar{L}_{-n},\qquad n>0.
\end{equation}
These expressions are in agreement with~\cite{Kausch} where
 the right-moving Virasoro generators for a non-zero integrer $n$ are expressed as 
\begin{eqnarray}\label{Ln-def}
L_n=\sum_{m\in\oZ} \sfermp_{n-m}\sfermm_m
\end{eqnarray}
and the generators for the left-moving part are 
\begin{eqnarray}\label{barLn-def}
\bar{L}_n=\sum_{m\in\oZ} \bsfermp_{n-m}\bsfermm_m.
\end{eqnarray}
The left and right Virasoro algebras
of course commute, and the vacuum is annihilated by all non-negative modes. 

Similarly, we can show that the scaling limit of $H(n)$ for $n<0$
gives also the sum  $L_n +\bar{L}_{-n}$ of left and right Virasoro generators.
To cover the full Virasoro, we still need to get $L_n -\bar{L}_{-n}$. 

\medskip
It turns out that the corresponding lattice analogue of $L_n
-\bar{L}_{-n}$ is the Fourier equivalent of the momentum operator $P$
in~\eqref{lattmom}
\begin{equation*}
P(n)=\ffrac{i}{2}\sum_{j=1}^N e^{-iqj} [e_j,e_{j+1}], \qquad q=\ffrac{n\pi}{L}.
\end{equation*}
We first obtain expression for the commutator in terms of $\theta$-fermions,
\begin{equation*}
 [e_j,e_{j+1}] =  \ffrac{1}{N}\sum_{p_1,p_2}e^{ij(p_2-p_1)}( e^{-ip_1} - e^{ip_2})(i-e^{-ip_1})(1+ ie^{ip_2})
\fermd_{p_1}\ferm_{p_2},\qquad 1\leq j\leq N,
\end{equation*}
which we use to get
\begin{equation*}
P(n) = 4e^{iq}\Bigl(\sum_{\substack{p=\step\\\text{step}=\step}}^{\pi-\step}
\Bigl(\cos{\bigl(p+\ffrac{q}{2}\bigr)}\sin{\ffrac{p+q}{2}}
\cos{\ffrac{p}{2}}\,\fermd_{p+\frac{\pi}{2}}\ferm_{p+q+\frac{\pi}{2}} +
[p\to p+\pi]\Bigr) +
\cos{\ffrac{q}{2}}\sin{\ffrac{q}{2}}\,\fermd_{\frac{\pi}{2}}\ferm_{q+\frac{\pi}{2}}
\Bigr), \quad q={\ffrac{n\pi}{L}}.\label{inter-mom}
\end{equation*}
We consider the case $0< n< L$.
Using the formulas~\eqref{eq:ferm-chi-eta} expressing the
$\theta$-fermions in terms of the $\chi$-$\eta$ fermions, we rewrite the $P(n)$ as
\begin{multline*}
P(n) = 2e^{iq}\Biggl(\cos{\ffrac{q}{2}}\sqrt{\sin{q}}\, \chi^{\dagger}_{0}\bigl(\chi_{q} -
\eta_{q}\bigr) + \sum_{\substack{p=\step\\\text{step}=\step}}^{\pi-q-\step}
\cos{\bigl(p+\ffrac{q}{2}\bigr)}\sqrt{\sin{(p)}\sin{(p+q)}} \bigl(\chi^{\dagger}_{p}\,\chi_{p+q} +
\eta^{\dagger}_{p}\,\eta_{p+q}\bigr)\\ 
- \sum_{\substack{p=\step\\\text{step}=\step}}^{q-\step}
\cos{\bigl(p-\ffrac{q}{2}\bigr)}\sqrt{\sin{(p)}\sin{(q-p)}}\bigl(\chi^{\dagger}_{\pi-p}\,\eta_{q-p} -
\eta^{\dagger}_{\pi-p}\,\chi_{q-p}\bigr)
- \cos{\ffrac{q}{2}}\sqrt{\sin{q}}\, \bigl(\chi^{\dagger}_{\pi-q} -
\eta^{\dagger}_{\pi-q}\bigr)\eta_{0}\Biggr)
\end{multline*}
which finally gives in the scaling limit (for any finite
mode $n$)
\begin{equation*}
\ffrac{L}{2\pi}P(n) \mapsto \sum_{m\in\oZ}\sfermp_{-m+n}\sfermm_{m} -
\sum_{m\in\oZ}\bsfermp_{-m-n}\bsfermm_{m} = 
L_{n}-\bar{L}_{-n},\qquad n>0.
\end{equation*}
We can similarly show that the scaling limit of $P(n)$ for $n<0$
gives also  $L_n -\bar{L}_{-n}$.

\subsection{The twisted model}\label{sec:twist-lim}
We can perform the same analysis in the model with antiperiodic
$\gl(1|1)$ fermions discussed in Sec.~\ref{sec:twist-mod}. This requires the introduction of a new set of momenta replacing   (\ref{momenta-set}):
\begin{equation}\label{momenta-set-R}
p_m=
\begin{cases}
\frac{2\pi m}{N},\qquad &L-\text{odd},\\
\frac{(2m-1)\pi}{N}, &L-\text{even},
\end{cases} 
\end{equation}
with, as before, $1\leq m\leq N$, while the formal expression~\eqref{def-ferm} for the fermions $\ferm_{p_m}$ and $\fermd_{p_m}$ is the same. Proceeding, we now find the Hamiltonian $H_{\ramond}$ in the antiperiodic model as
\begin{equation}
H_{\ramond}=2\sum_p (1+\sin p)\theta_p^\dagger\theta_{p+\pi}
\end{equation}
which is the same formal expression as for the periodic model. The
difference is that now the momenta run over a different set. As a
result, the values $p={\pi\over 2}, {3\pi\over 2}$ are not allowed, and
there are no zero modes. The ground state in this model is non
degenerate, and we find
\begin{equation}
\langle \hbox{vac}|H_{\ramond}|\hbox{vac}\rangle=-2\bigl(\sin{\ffrac{\pi}{N}}\bigr)^{-1}=-\ffrac{2N}{\pi}-\ffrac{\pi}{3N}+o\,(1/N)
\end{equation}
which corresponds to an effective central charge
$c_{\text{eff}}=1=-2-24\times {-1\over 8}$. We introduce
$\chi^{(\dagger)}_p$ and $\eta^{(\dagger)}_p$ fermions generating
Hamiltonian eigenstates from the vacuum by the same
formal definition~\eqref{eq:chi-eta-def}  but
now momenta takes values $\step/2\leq p\leq \pi - \step/2$ with the
step $\step=\pi/L$. The normal
ordered Hamiltonian then reads
\begin{equation}
H_{\ramond}=2\sum_{\substack{p=\step/2\\\text{step}=\step}}^{\pi-\step/2}\sin{p}\bigl(\chi^{\dagger}_p\chi_{p}
+ \eta_p\eta^{\dagger}_{p}\bigr) +\langle \hbox{vac}|H_{\ramond}|\hbox{vac}\rangle
\end{equation}
with $\step={2\pi\over N}$, and the momenta are of the
form\footnote{Like in the periodic model, the difference with
${\pi\over 2}$ in the notation for $\chi_p$, $\eta_p$ fermions and
$\ferm_p$ fermions makes both cases $L$ even and odd similar.}
$p=\frac{(2m-1)\pi}{N}$, with $1\leq m\leq L$. Introducing exactly the
same definition  for modes as in~\eqref{eq:sferm-lat-def}, with $p_m=(m-1/2){\pi\over L}$, gives the
scaling limit
\begin{equation}
\ffrac{L}{2\pi}\Bigl(H_{\ramond}+\ffrac{2N}{\pi}\Bigr)\mapsto L_0+\bar{L}_0-\ffrac{1}{12},
\end{equation}
with the representarion of the Virasoro modes now 
\begin{equation}
L_0+\bar{L}_0=\sum_{m\in\oZ} :\psi_{-m+1/2}^2\psi_{m-1/2}^1: + \bigl(\psi\to\bar{\psi}\bigr).
\end{equation}
Similar analysis of the Hamiltonians $H(n)$ and of the momenta modes $P(n)$ provides the expected formulas for $L_n$ and $\bar{L}_n$ in this case as well.

\subsection{From $\JTL{N}$ to $\Vir(2)\boxtimes\overline{\Vir}(2)$?}\label{sec:fromJTL-Vir}

It is possible to calculate the scaling limit of more complicated
expressions. In particular, it is known that the scaling limit of the
logarithm of the transfer matrix itself involves only $L_0$ and
$\bar{L}_0$. Expanding this transfer matrix in powers of the spectral
parameter shows that there is an infinity of lattice Hamiltonians (see below for more details) and
momenta with identical scaling limits~\cite{KooSaleur}. For instance,
instead of taking $H\propto -\sum e_i$ we could take the next
Hamiltonian $H\propto \sum [e_j,[e_{j+1},e_{j+2}]]$, which should also
give $L_0+\bar{L}_0$ when acting on low energy states\footnote{More complicated expressions in the enveloping algebra of the Virasoro algebra would be obtained if one were to retain terms of higher order in $1/N$.  This is discussed in~\cite{KooSaleur}.}. This shows that
the correspondence between $\JTL{N}$ elements and elements in the product $\VirN(2)$ of
left and right Virasoro algebras is certainly not a bijection.  

While most of the foregoing results (such as the existence of expressions in JTL generators which converge in the scaling limit to Virasoro generators) are expected to hold for more general models,
how this precisely occurs 
 is not fully understood in general, because of our only partial control on the eigenstates of the Hamiltonian and matrix elements of generators (through the algebraic Bethe ansatz). Even the fact that
the Fourier modes of the local density of energy and momentum give,
when restricting to low lying energy states, the modes of the stress
energy tensor, can only be established analytically in free fermionic
models -- the Ising chain~\cite{KooSaleur}, and the $\gl(1|1)$ chain
here.

Indeed, a major difficulty in studying the correspondence between lattice algebras  and  $\VirN(2)$ is that  the lattice algebra acts  on all the states of the lattice model, including a priori  the high energy states  which disappear in the scaling limit. As a result, it is not clear on general grounds how to relate the structure of   $\JTL{N}$ modules and  $\VirN(2)$ modules: for instance,  we could have two $\JTL{N}$ modules in the spin chain mapped  by some words in  $\JTL{N}$ generators, but in such a way that this connection involves only highly excited states, and disappears when we restrict to excitations at small momentum and energy. On the other hand, it is tempting to speculate in general that low and high energy states are not special in an algebraic sense, so that, if a mapping exists between two modules (subquotients), it will still be present when restricting to the scaling limit. 

Of course, for  $\gl(1|1)$ things are  particularly simple: a look at $H(n)$ in~\eqref{lat-Vir-H} for instance shows that, for any finite $n$ as $L$ becomes large, it only connects low energy states to low energy states and high energy states to high energy states. This implies that the continuum limit of products of $H(n)$'s should coincide with the product of their continuum limits --- in particular, we can easily compute the commutators
\begin{equation*}
\bigl[H(n),H(-n)\bigr] = -4\sin q \sum_p \sin(2p) \fermd_p\ferm_p = -4\sin (q) P,\qquad q=\ffrac{n\pi}{L},
\end{equation*}
using the finite-chain fermionic expression~\eqref{inter}, and their scaling limit $\bigl[\ffrac{L}{2\pi}H(n),\ffrac{L}{2\pi}H(-n)\bigr]\mapsto 2n(L_0-{\bar L}_0)$. On the other hand, the commutator of the scaling limits~\eqref{scal-lim-Hn} of $\ffrac{L}{2\pi}H(\pm n)$ gives the same expression.  One can then for instance obtain the central charge directly from the commutator $[H(n),P(-n)]$. Indeed, a long calculation gives
\begin{equation}\label{Hn-Pn-comm}
\left [\ffrac{L}{2\pi}H(n),\ffrac{L}{2\pi}P(-n)\right] = \ffrac{L^2}{2\pi^2} e^{-iq/2}\sin (q/2)\Bigl(\sin^2\ffrac{q}{2}\,H[\cos^2] - \cos^2\ffrac{q}{2}\bigl(H-3H[\sin^2]\bigr)\Bigr),
\end{equation}
where the Hamiltonian $H$ is given in~\eqref{eq:hamilt-fermd} and we use the notation  $H[f]=2\sum_p f(p)(1+\sin p)\fermd_p\ferm_{p+\pi}$ for  Hamiltonians modified by a weight $f(p)$, where $f$ is  a periodic function $f(p+2\pi)=f(p)$. For $f(p)=\cos^2(p)$, we have the normal-ordered expression (in terms of the $\chi$-$\eta$ fermions introduced above)
\begin{equation}\label{H-cos-chi-eta}
H[\cos^2] = 2\sum_{p=\step}^{\pi-\step}\sin^3 p \bigl(\chi^{\dagger}_p\chi_{p}
+ \eta_p\eta^{\dagger}_{p}\bigr) - 2\sum_{p=\step}^{\pi-\step}\sin^3{p},
\end{equation}
where we have extracted the ground-state value of $H[\cos^2]$ in the second sum (compare with~\eqref{eq:ham-chi-eta-norm}), which has the leading asymptotic for large $L$
\begin{equation}\label{H-cos-vac-val}
\vacl H[\cos^2]\vacr = -2\sum_{m=1}^{L-1}\sin^3{\ffrac{m\pi}{L}} = \half\bigl(\cot{\ffrac{3\pi}{N}}-3\cot{\ffrac{\pi}{N}}\bigr) =
-\ffrac{4N}{3\pi} + o\,(1/N),
\end{equation}
with an $N$-linear contribution  canceled.
We then note $H[\sin^2]=H-H[\cos^2]$ and that the first sum in~\eqref{H-cos-chi-eta} give a contribution  of order $1/L^2$ to the Hamiltonian $H$ in
\eqref{Hn-Pn-comm} which has to be neglected in the scaling limit. We thus keep only the vacuum value~\eqref{H-cos-vac-val} to obtain finally the scaling limit of~\eqref{Hn-Pn-comm}
\begin{equation*}
\left [\ffrac{L}{2\pi}H(n),\ffrac{L}{2\pi}P(-n)\right]  \mapsto 2n (L_0+{\bar L}_0) + \ffrac{c}{6}n(n^2-1) = \left[L_n+\bar L_{-n}, L_{-n}-\bar L_{n}\right].
\end{equation*}

A similar calculation using the fermions shows that all other products also commute with the scaling
 limit, so that in particular the scaling limit of a commutator is the
 commutator of the scaling limits.

\newcommand{\specp}{\mathsf{u}}
\subsubsection{Higher Hamiltonians and their Fourier images}
It is also interesting (and we will use these results in our subsequent papers) to consider the scaling limit of the
whole family of higher Hamiltonians in the periodic $\gl(1|1)$
spin-chain. These can be obtained using the underlying integrable structure, and building the  family of commuting diagonal-to-diagonal transfer matrices
$T_d(\specp)$. An expansion of (the logarithm of) $T_d(\specp)$ in powers of  $\specp$  produce an infinite of commuting operators  $H_l(0)$, with $l\geq0$, see~\cite{KooSaleur} and references therein. To explore the properties of these $H_l(0)$,  we first compute multiple
commutators of the $\JTL{N}$-generators $e_j$
\begin{equation}
E_{j,l} = \Bigl[e_j,\bigl[e_{j+1},\dots
 e_{j+l-2},[e_{j+l-1},e_{j+l}]\dots\bigl]\Bigl], \qquad 1\leq j\leq N.
\end{equation} 

By an induction, we prove the following, for $1\leq j\leq N$ and $l>0$,
\begin{equation}
E_{j,l} = 
\begin{cases}
(-1)^j\ffrac{i}{N}\sum_{p_1,p_2}e^{ij(p_2-p_1)}( e^{ilp_2} + e^{-ilp_1})(i-e^{-ip_1})(1+ ie^{ip_2})
\fermd_{p_1}\ferm_{p_2},\qquad &l - \text{even},\\
-\ffrac{1}{N}\sum_{p_1,p_2}e^{ij(p_2-p_1)}( e^{ilp_2} - e^{-ilp_1})(i-e^{-ip_1})(1+ ie^{ip_2})
\fermd_{p_1}\ferm_{p_2},\qquad &l - \text{odd},
\end{cases}
\end{equation}
where the sums are taken over all allowed momenta $p_1$, $p_2$ from
the set~\eqref{momenta-set}. Then, the integrable Hamiltonians
$H_l(0)$ are given by the sums of the $E_{j,l}$ over all sites. In
particular, the operators $H_0(0)=H$ and $H_1(0)=P$ were studied
above in Sec.~\ref{sec:emerge-latVir} where we also studied their
Fourier images.  To find fermionic expressions for
Fourier images of all the higher Hamiltonians
\begin{equation*}
H_l(n) = -\half e^{-il\frac{\pi}{2}}\sum_{j=1}^Ne^{-iqj}E_{j,l},\qquad q=\ffrac{n\pi}{L} \quad\text{and}\quad l\in\oN,\quad n\in\oZ,
\end{equation*}
we repeat all the previous steps in the study of $H(n)$ and $P(n)$ in
Sec.~\ref{sec:emerge-latVir} and get 
\begin{equation*}
H_l(n) = 
4 e^{iq(l+1)/2}\sum_{p=0}^{2\pi-\step} \cos\bigl(l(p+q/2)\bigr)\cos\ffrac{p}{2}
\begin{cases}
\cos\ffrac{p+q}{2}\,
\fermd_{p+\frac{\pi}{2}}\ferm_{p+q-\frac{\pi}{2}},&\quad
l-\text{even},\\
\sin\ffrac{p+q}{2}\,
\fermd_{p+\frac{\pi}{2}}\ferm_{p+q+\frac{\pi}{2}},&\quad l-\text{odd},
\end{cases}
\end{equation*}
which we rewrite in terms of  the $\chi$-$\eta$ fermions, for integer $0\leq n < L$, as
\begin{multline*}
H_l(n) =
2e^{iq\frac{l+1}{2}}\Biggl( \sum_{\substack{p=\step\\\text{step}=\step}}^{\pi-q-\step}
\cos{l\bigl(p+\ffrac{q}{2}\bigr)}\sqrt{\sin{(p)}\sin{(p+q)}} \bigl(\chi^{\dagger}_{p}\,\chi_{p+q}
- (-1)^l\eta^{\dagger}_{p}\,\eta_{p+q}\bigr)\\
+\sum_{\substack{p=\step\\\text{step}=\step}}^{q-\step}
\cos{l\bigl(p-\ffrac{q}{2}\bigr)}\sqrt{\sin{(p)}\sin{(q-p)}}\bigl(\eta^{\dagger}_{\pi-p}\,\chi_{q-p} +
 (-1)^l\chi^{\dagger}_{\pi-p}\,\eta_{q-p}\bigr)\\
+\cos{\ffrac{lq}{2}}\sqrt{\sin{q}}\bigl(\chi^{\dagger}_{0}\chi_{q} +(-1)^l \chi^{\dagger}_{0}\eta_{q} + \eta^{\dagger}_{\pi-q}\eta_{0} +
 (-1)^l\chi^{\dagger}_{\pi-q}\eta_{0} \bigr) +
\delta_{n,0}\bigl(1+(-1)^l\bigr)\chi^{\dagger}_{0}\eta_{0} \Biggr).
\end{multline*}
This finally gives in the scaling limit (for finite $n$ and $l$) the
left and right Virasoro generators:
\begin{equation}\label{Hln-scal-lim}
\ffrac{L}{2\pi}H_l(n) \mapsto \sum_{m\in\oZ}\sfermp_{-m+n}\sfermm_{m} +(-1)^l
\sum_{m\in\oZ}\bsfermp_{-m-n}\bsfermm_{m} = 
L_{n}+(-1)^l\bar{L}_{-n}, \qquad l\geq0,
\end{equation}
which does not depend on $l$, only on its value modulo $2$. We note that this result is  obtained 
by  taking the leading term in the expansion
$\cos{l\bigl(p\pm\ffrac{q}{2}\bigr)}=1-\frac{(m\pm n/2)^2}{L^2}\pi^2
l^2+\dots$ only.
It is interesting to explore the content of the higher order terms in the scaling limit, and their relation 
with  conserved quantities in the conformal field theory.   We leave this
  problem for a future work \cite{GST}.
A very similar calculation gives  the same
 scaling limit~\eqref{Hln-scal-lim} for all negative modes $n<0$ as well.

To examine further the relation between $\JTL{N}$ and
$\Vir(2)\boxtimes\overline{\Vir}(2)$, it is possible to compare the
modules over these two algebras present respectively in the spin chain and the
continuum limit. This will be discussed in our third
paper~\cite{GRS3}. But before launching into representation theory, a lot can be
learned from the analysis of the lattice symmetries, to which we now
return.

\section{Symmetries and the scaling limit}\label{sec:symm-continuum}

The expectation that  the natural equivalent of the $\JTL{N}$ algebra in
the continuum limit
would be the product  of
the left and right Virasoro algebras encounters difficulties when we consider the centralizer~$\centVir$ of $\VirN(2)=\Vir(2)\boxtimes\overline{\Vir}(2)$. 
While for finite chains, the centralizer of
$\JTL{N}$ is $\LQGodd$, in the continuum limit, it is well known that
 $\VirN(2)$ commutes at
least with $\gl(1|1)$ and an $\nSU(2)$ symmetry discovered by
Kausch. The situation in the  boundary and periodic cases is thus quite different.
Several questions  arise as a result, the most obvious being, what
happens to $\LQGodd$ and how it is related with the continuum
$\nSU(2)$. This is what we consider first.

\medskip

 We first introduce the continuum fermions  via the mode expansion in the complex plane~\cite{Kausch}
\begin{equation}\label{sympl-ferm-def}
\Phi^\alpha(z,\bar{z})=\phi_0^\alpha-i\psi_0^\alpha\ln(z\bar{z})+i\sum_{m\neq 0}\ffrac{\psi^\alpha_m}{m}z^{-m}+
\ffrac{\bar{\psi}_m^\alpha}{m}\bar{z}^{-m}, \qquad \alpha,\beta\in\{1,2\},
\end{equation}
where the modes have the anti-commutation relations
\begin{equation*}
\{\psi^{\alpha}_m,\psi^{\beta}_{m'}\} =
mJ^{\alpha\beta}\delta_{m+m',0}\,,\qquad
\left\{\phim_0,\sfermp_0\right\}=i,\qquad
\left\{\phip_0,\sfermm_0\right\}=-i,
\end{equation*}
with the symplectic form $J^{12}=-J^{21}=1$.
Then, the generators of the global $\nSU(2)$ in the symplectic-fermion theory are
\begin{equation}\label{eq:Kausch-SU}
Q^a=d^a_{\alpha\beta}\left\{i\phi^\alpha_0\psi^\beta_0+\sum_{n=1}^\infty
\left(\ffrac{\psi^\alpha_{-n}\psi^\beta_n}{n} + \ffrac{\bar{\psi}^\alpha_{-n}\bar{\psi}^\beta_n}{n}\right)\right\}
\end{equation}
with 
\begin{equation}
d^0_{\alpha\beta}=
\half
\begin{pmatrix}
-1 & 0\\
0 & -1
\end{pmatrix}, \quad
d^1_{\alpha\beta}=
\half
\begin{pmatrix}
1&0\\
0&-1
\end{pmatrix}, \quad
d^2_{\alpha\beta}=
\half
\begin{pmatrix}
0&-1\\
-1&0\end{pmatrix},
\end{equation}
with $[Q^a,Q^b]=f^{ab}_cQ^c$ and $f^{01}_2=-1$.

A superficial look at the model would suggest that this $\nSU(2)$
should somehow `emerge' from the lattice symmetry $\LQGodd$.  In the
open case indeed, the lattice model has the full quantum group $\LQG$
symmetry, and the $\SU(2)$ part in that case coincides with the
$\SU(2)$ of the continuum limit. While in the periodic case, the
lattice model has less symmetry, the degeneracies remain of the form
$4j$, see~\cite{GRS2}, and would suggest that the non-commutation with Temperley--Lieb
of the `even part' of $\LQG$ is a lattice effect disappearing in the
continuum limit. A more careful look at the model shows that this
expectation is not correct at all. Maybe the quicker is simply to work out
the scaling limit of the generators.

\subsection{Scaling limit of the  (lattice) $s\ell(2)$ generators $\e$ and
  $\f$}
We consider therefore the scaling limit of the $s\ell(2)$ generators --
the renormalized powers $\e$ and
  $\f$. These do not commute with $\JTL{N}$ or even the Hamiltonian for a finite lattice, and the question is, how they related to the  $\nSU(2)$ generators in the continuum limit.  Recall  first the fermionic formula~\eqref{QG-ferm-2} for the operator $\e$,
\begin{multline}
\e =-\sum_{\substack{p=\step\\\text{step}=\step}}^{\pi-\step}
\cot\ffrac{p}{2}\,\fermd_{p+\frac{\pi}{2}}\,\fermd_{\frac{\pi}{2}-p} 
-
2i\sum_{p\ne\pi}\frac{\fermd_{p-\frac{\pi}{2}}\,\fermd_{\frac{\pi}{2}}}{e^{ip}+1}\\
= i\fermd_{\frac{\pi}{2}}\,\fermd_{\frac{3\pi}{2}} 
+ \sum_{\substack{p=\step\\\text{step}=\step}}^{\pi-\step}
\Bigl(\cot\ffrac{p}{2}\,\fermd_{\frac{\pi}{2}-p}\,\fermd_{p+\frac{\pi}{2}}
+
\ffrac{e^{-ip/2}}{\cos{p/2}}\bigl(\cot{\ffrac{p}{2}}\fermd_{p+\frac{\pi}{2}}
- i\fermd_{p-\frac{\pi}{2}}\bigr)\fermd_{\frac{\pi}{2}}\Bigr),
\end{multline}
which can be rewritten in the $\chi$-$\eta$ notation as
\begin{equation}\label{eq:div-e-chi-eta}
\e = \phip_0\sfermp_0 +
\sum_{p=\step}^{\pi-\step}
\eta^{\dagger}_p\chi^{\dagger}_{\pi-p} +
\sum_{p=\step}^{\pi-\step}
\ffrac{e^{-ip/2}}{\sqrt{\sin{p}}} \bigl((1-i)\chi^{\dagger}_p
-(1+i)\eta^{\dagger}_p\bigr) \chi^{\dagger}_0,
\end{equation}
where we introduce the operators $\phi^{1,2}_0$ conjugated to $\psi_0^{1,2}$,
\begin{eqnarray}
\phip_0=-i\sqrt{\pi/L}\,\fermd_{\frac{3\pi}{2}}, \qquad
\phim_0=i\sqrt{\pi/L}\,\ferm_{\frac{\pi}{2}}\\
\left\{\phim_0,\sfermp_0\right\}=i,\qquad
\left\{\phip_0,\sfermm_0\right\}=-i.
\end{eqnarray}
Going to the
$\psi^{\alpha}$-fermions defined in~\eqref{eq:sferm-lat-def} 
 gives the scaling limit
\begin{equation}
\e \mapsto  \phip_0\sfermp_0 + \sum_{m>0}
\Bigl(\ffrac{\sfermp_{m}\sfermp_{-m}}{m} -
\ffrac{\bsfermp_{m}\bsfermp_{-m}}{m}\Bigr)
+ \sum_{m\ne0} \ffrac{1}{m}\Bigl((i+1)\bsfermp_{m}\bsfermp_{0} +
(i-1)\sfermp_{m}\sfermp_{0}\Bigr).\label{elinear}
\end{equation}
The scaling limit for  $\f$ is given by similar formula with the
substitution $\psi^2\to \psi^1$, $\phip\to \phim$. 

As we can see, the scaling limit  of the renormalized powers $\e$ and
$\f$ describes a  different $s\ell(2)$ than the global $\nSU(2)$
we have in the symplectic fermions theory mostly
because of the second sum. Mainly for these reasons, the
four-dimensional space of the ground
states~\eqref{diag:space-gr-st}, spanned by the vacuum $\vac$, the state
 $\lvac$ and the two fermionic states $\phi^{1,2}$,
is \textit{not} invariant under the action of $\e$  and $\f$ on a
finite lattice. Indeed, it is easy to check using~\eqref{eq:div-e-chi-eta} that the vacuum
$\vac$ is the $s\ell(2)$-invariant while its logarithmic partner $\lvac$ is not an invariant,
\begin{equation*}
\e(\lvac) =
\sum_{m=1}^{L-1}(-1)^{m-1}(1-i)\ffrac{e^{-i\frac{\pi m}{2L}}}{\sqrt{\sin{\ffrac{m\pi}{L}}}}\prod_{\substack{j=1\\\,j\ne
m}}^{L-1}\chi_{\frac{j\pi}{L}}\vectv{\uparrow\dots\uparrow} \equiv
\lvac'\qquad \text{and}\qquad
\e(\phim) = i\phip + \eta_0\lvac',
\end{equation*}
This  is not surprising because the Hamiltonian on a finite lattice
does not commute with the $s\ell(2)$ generated by the $\e$ and $\f$. We see therefore that the natural $s\ell(2)$ generators obtained from $\LQG$ bear no simple relationship with Kausch's 
$\nSU(2)$ in the periodic case. 

\subsection{Scaling limit of $\LQGodd$ in the periodic model}\label{sec:scal-lim-centJTL}
The additional elements $\f^L$ and $\e^L$ which do not belong to
$\LQGodd$ but commute with $\JTL{2L}$  have no meaning
in the scaling limit $L\to\infty$ (they are non zero only on extremely excited states that are not part of that limit)  and we thus suppress them and make
no difference between $\centJTL$ and $\LQGodd$.

In the scaling limit, the centralizer $\LQGodd$ gives rise to the zero modes $\F\mapsto\psi_0^1$ and $\E\mapsto\psi_0^2$ and products of these with the renormalized even powers. We thus get, in the limit, the generators
%
\begin{equation}\label{scal-lim-LQGodd}
  \EE n =\e^n\E \mapsto 
\left[\sum_{m>0}
\Bigl(\ffrac{\sfermp_{m}\sfermp_{-m}}{m} -
\ffrac{\bsfermp_{m}\bsfermp_{-m}}{m}\Bigr)\right]^n\psi_0^2,\qquad
\FF n =\f^n\F \mapsto 
\left[\sum_{m>0}
\Bigl(\ffrac{\sfermm_{m}\sfermm_{-m}}{m} -
\ffrac{\bsfermm_{m}\bsfermm_{-m}}{m}\Bigr)\right]^n\psi_0^1,
\end{equation}
The Cartan element $\h$ meanwhile is given on a finite lattice by
\begin{equation}
2\h =  S^z = \sum_p\fermd_p\ferm_p - L  =
\sum_{p=0}^{\pi-\step} \bigl(\chi^{\dagger}_p\chi_p - \eta_p\eta^{\dagger}_p\bigr)   
\end{equation}
and has the limit
\begin{equation}\label{scal-lim-h}
2\h\mapsto - i \bigl(\sfermp_0\phim_0 + \sfermm_0\phip_0\bigr) + \sum_{m>0}
\ffrac{1}{m}\bigl(\sfermp_{-m}\sfermm_{m} + \sfermm_{-m}\sfermp_{m} + [\psi\to\bar{\psi}]\bigr)
\end{equation}
while the generator $\K = (-1)^{2\h}$.  Note that the value of $S^z$ on
the lattice is {\it twice} the value of the third component of the
$\nSU(2)$ isospin in the continuum. The case $S^z$ even (odd)
corresponds to bosonic (fermionic) states, so the continuum isospin is
integer (respectively, half integer).

Using the symplectic fermions expressions~\eqref{Ln-def}
and~\eqref{barLn-def} for the left and right Virasoro modes $L_n$,
$\bar{L}_n$, we see that the scaling limit~\eqref{scal-lim-LQGodd}
and~\eqref{scal-lim-h}
of the $\JTL{N}$ centralizer $\centJTL$ does commute with the full
Virasoro algebra $\VirN(2)$ (the multiplication by the zero modes suppresses all the unwanted terms
in the expression \eqref{elinear}.)
We should also note that the limit of $\centJTL$ cannot be obtained
as the multiplication of the global $\nSU(2)$ with the zero modes.
 There remains  indeed a different 
 sign between the left and right moving components in the two expressions, meaning once again that the lattice objects identified so far 
are not related with the Kausch's $\nSU(2)$.

\subsection{How to get  the (continuum) $\nSU(2)$ generators from the spin
  chain?}\label{sec:SUgen-lat} 

Of course, it is possible  to study in more detail the scaling limit
 of the lattice fermions themselves, and thus build somewhat artificially lattice  quantities which are not symmetries of the problem in finite size, but go over to the  $\nSU(2)$ generators in the continuum limit. 
A little trial and error suggests the introduction of  
\begin{equation}\label{eq:epl-fpl-def}
\te{+} = \chi_0^{\dagger}\eta_0^{\dagger}-\sum_{p=\step}^{\pi-\step} (\cos{p})\; \eta^{\dagger}_p\chi^{\dagger}_{\pi-p},\qquad
\tf{+} =  \eta_0\chi_0 +
\sum_{p=\step}^{\pi-\step} (\cos{p})^{-1}\;\chi_p\eta_{\pi-p},
\end{equation}
which look like~\eqref{eq:div-e-chi-eta} but are  slightly modified by
the introduction of the  weight $\cos{p}$
\begin{equation}\label{eq:tildee-fermd}
\te{+} = \half\sum_p (\sin{p})\; \fermd_{p}\fermd_{-p},\qquad
\tf{+} = \half\sum_{p\ne0,\pi} (\sin{p})^{-1}\; \ferm_{-p}\ferm_{p},\qquad\text{with} \qquad [\te{+},\tf{+}]=S^z.
\end{equation}
We now have 
\begin{equation}\label{ep-fp-lin}
\te{+}\mapsto -i \phip_0\sfermp_0 +
\sum_{m>0}
\Bigl(\ffrac{\sfermp_{m}\sfermp_{-m}}{m} +
\ffrac{\bsfermp_{m}\bsfermp_{-m}}{m}\Bigr),\qquad
\tf{+}\mapsto i \phim_0\sfermm_0 +
\sum_{m>0}
\Bigl(\ffrac{\sfermm_{-m}\sfermm_{m}}{m} +
\ffrac{\bsfermm_{-m}\bsfermm_{m}}{m}\Bigr)
\end{equation}
in agreement with the expressions~\eqref{eq:Kausch-SU} of the global
 $\nSU(2)$ generators.
 It is straightforward to check that, on a finite-lattice, the $\te{+}$ and $\tf{+}$ 
commute with the Hamiltonian~\eqref{eq:hamilt-fermd} which can be easily checked
using~\eqref{eq:tildee-fermd}:
\begin{equation*}
[\te{+},H] = \sum_p \sin{p}(1+\sin{p}) \fermd_{\pi-p}\fermd_{p} = 0.
\end{equation*}
However, $\te{+}$ and $\tf{+}$  {\sl do not commute} with $H(n)$ for $n\geq1$ -- and thus are not part of $\centJTL$.
The reader interesting in the centralizer $\cent_{H}$ of the Hamiltonian (but not of the whole algebra $\JTL{N}$) can find a discussion in Sec.~\ref{sec:cent-H} below.

\begin{rem}
We could equivalently study the family of operators (generalizing~\eqref{eq:epl-fpl-def})
\begin{equation}\label{ten-tfn-def}
\te{n} = \chi_0^{\dagger}\eta_0^{\dagger}-\sum_{p=\step}^{\pi-\step} (\cos{p})^{n}\; \eta^{\dagger}_p\chi^{\dagger}_{\pi-p},\qquad
\tf{n} =  \eta_0\chi_0 +
\sum_{p=\step}^{\pi-\step} (\cos{p})^{-n}\;\chi_p\eta_{\pi-p}, \qquad n - \text{odd},
\end{equation}
with the $s\ell(2)$ relations
\begin{equation}\label{eq:lat-min-glSU-comm}
[\h,\te{n}] = \te{n},\qquad [\h,\tf{n}] = -\tf{n},\qquad
[\te{n},\tf{n}]=2\h = S^z.
\end{equation}
The
scaling limit of $\te{n}$ and $\tf{n}$  are all identical with~\eqref{eq:Kausch-SU}, but   we stress that these operators do not
commute with the $\JTL{N}$.
The $\te{n}$ and $\tf{n}$ are however in the centralizer for
the Hamiltonian $H(0)$, which is easy to check  -- see also Sec.~\ref{sec:cent-H} below for more details.

In terms of  $\theta$-fermions expression the generators read
\begin{equation}
\te{n} = \half\sum_p (\sin{p})^{n}\; \fermd_{p}\fermd_{-p},\qquad
\tf{n} = \half\sum_{p\ne0,\pi} (\sin{p})^{-n}\; \ferm_{-p}\ferm_{p}.
\end{equation}
Going back to real space however leads to a strongly non local expression for one of these generators ({\it e.g.}, $\tf{n}$ for $n$ positive), since the pole in the Fourier transform 
give rise to a power law growth for the couplings between pairs of fermions  $f_j$.

\end{rem}

\subsection{The twisted model}

In the model with anti-periodic boundary conditions introduced and studied in Sec.~\ref{sec:twist-mod}, things are a bit different. There are no zero modes, and the continuum limit of the $U \SU(2)$ (the centralizer of the $\JTL{N}^{tw}$) generators reads simply 
\begin{equation}
\tilde{Q}^a=d^a_{\alpha\beta}\sum_{n=0}^\infty
\left(\ffrac{\psi^\alpha_{-n-1/2}\psi^\beta_{n+1/2}}{n+1/2} - \ffrac{\bar{\psi}^\alpha_{-n-1/2}\bar{\psi}^\beta_{n+1/2}}{n+1/2}\right).
\end{equation}
Of course, in this case  the continuum limit exhibits in fact two
$\SU(2)$'s, left and right being fully factorized (while they remain
coupled by the zero modes in the periodic model). These two $\SU(2)$'s can
be combined with plus or minus sign; the lattice symmetry (see Sec.~\ref{sec:twist-mod}) becomes one
of them.

\subsection{Remarks about the  Hamiltonian centralizer and  loop $s\ell(2)$
  symmetry}\label{sec:cent-H}
It is interesting to consider further the  centralizer $\cent_{H}$ of
the Hamiltonian $H$ on a finite lattice.
For this, we first give the quantum-group expression for the $n=0$
member of the $s\ell(2)$ family~\eqref{ten-tfn-def}
\begin{equation}\label{epr-fpr-def}
\epr = \ffrac{1}{N}\bigl(\Er\E - i[\e\E,\Fr]\K^{-1}\bigr),\qquad 
\fpr = \ffrac{1}{N}\bigl(\F\Fr + i[\f\F,\Er]\K\bigr), 
\end{equation}
where $\Er$ and $\Fr$ are generators of (representation of)
$U_{\q^{-1}}s\ell(2)$ 
\begin{equation}
\Delta^{N-1}(\Er) =\q^{-1} \sum_{j=1}^{N}\q^{-j} c_j^{\dagger} \K^{-1} = \sqrt{N}\fermd_{3\pi/2}\K^{-1},
\qquad \Delta^{N-1}(\Fr) = \sum_{j=1}^{N}\q^{-j+1} c_j = \q\sqrt{N}\ferm_{\pi/2}.\label{QG-right-ferm-1}
\end{equation}
Recalling also the fermionic
expressions for the generators $\E$ and $\F$
in~\eqref{QG-ferm-1},  we obtain
\begin{eqnarray}
\epr = \fermd_{\frac{\pi}{2}}\fermd_{\frac{3\pi}{2}}
- \e\,\fermd_{\frac{\pi}{2}}\ferm_{\frac{\pi}{2}} -
\ferm_{\frac{\pi}{2}}\fermd_{\frac{\pi}{2}}\e =
\fermd_{\frac{\pi}{2}}\fermd_{\frac{3\pi}{2}} - 
\sum_{p=\frac{\pi}{2}+\step}^{\frac{3\pi}{2}-\step}
\tan{\ffrac{1}{2}\bigl(p+\ffrac{\pi}{2}\bigr)}\fermd_{p}\fermd_{\pi-p},\\
\fpr =
\ferm_{\frac{3\pi}{2}}\ferm_{\frac{\pi}{2}} +
 \f\,\ferm_{\frac{3\pi}{2}}\fermd_{\frac{3\pi}{2}} +
\fermd_{\frac{3\pi}{2}}\ferm_{\frac{3\pi}{2}}\f =
\ferm_{\frac{3\pi}{2}}\ferm_{\frac{\pi}{2}} +
\sum_{p=\frac{\pi}{2}+\step}^{\frac{3\pi}{2}-\step}
\cot{\ffrac{1}{2}\bigl(p+\ffrac{\pi}{2}\bigr)}\ferm_{\pi-p}\ferm_{p}.
\end{eqnarray}

It is then possible to  show that the $\epo$ and $\fpo$ together with
$\epl$ and $\fpl$  --
a lattice analogue of Kausch's $\nSU(2)$ defined
in~\eqref{eq:epl-fpl-def} -- generate a loop $s\ell(2)$ algebra. First, we note~\eqref{eq:tildee-fermd}, and~\eqref{eq:lat-min-glSU-comm} is true for $n=0$, and $[\epo,\epl]=[\fpo,\fpl]=0$. Then, we only need to check the higher-order Serre relations
\begin{gather}
[\epo,[\epo,[\epo,\fpl]]] = [\fpo,[\fpo,[\fpo,\epl]]] = 0,\label{loop-sl-Serre-1}\\
[\epl,[\epl,[\epl,\fpo]]] = [\fpl,[\fpl,[\fpl,\epo]]] = 0.\label{loop-sl-Serre-2}
\end{gather}
Using~\eqref{eq:epl-fpl-def}, we compute the double commutators
$[\epo,[\epo,\fpl]] = -2\te{-1}$, {\it etc.},
which immediately give the Serre relations~\eqref{loop-sl-Serre-1} and we proceed similarly to get~\eqref{loop-sl-Serre-2}. 

Since we have seen that the $\te{n}$ and $\tf{n}$ commute with the
Hamiltonian, we have thus found a loop algebra symmetry \textit{of the
Hamiltonian} $H$ in the $\gl(1|1)$ spin chain. This is much like the
symmetry uncovered in~\cite{DFMC,Deguchi}, but a more careful
comparison shows that the sectors we are considering are different:
while in~\cite{DFMC}, the loop algebra is observed for periodic
(antiperiodic) XX spin chain and even (odd) spin, ours is obtained in
the opposite case, corresponding to periodicity for the $\gl(1|1)$
fermions. We stress that in contrast with the main focus of this
paper, the loop algebra is only a symmetry of the Hamiltonian, and
does not extend to the full $\JTL{N}$ algebra.

Of course, having observed the loop algebra on the lattice it is
natural to ask what happens of it in the continuum limit.  We have
already seen in~\eqref{ep-fp-lin} that the scaling limit of $\epl$ and $\fpl$ coincides
with the $\SU(2)$ generators (\ref{eq:Kausch-SU}). The scaling limit
of $\epr$, $\fpr$ gives very similar expressions, only with the
opposite sign between the chiral and
  antichiral components in the sum. In the end, we get a
(representation of the) loop $s\ell(2)$ algebra, with further additional
relations  like
$[\epo,[\epo,\fpl]] = -2\epl$ due to coincidence of $\epl$ with $\te{-1}$ and $\fpl$ with
$\tf{-1}$ in the scaling limit, and in the leading order.

We note however that  there exists a potential for yet more symmetries
of the Hamiltonian in the finite-lattice problem. Indeed, while the loop $s\ell(2)$ describes
intertwining operators of the Hamiltonian between sectors $\Hilb_{[j]}$
and $\Hilb_{[j']}$, with $|j-j'|=0\mod2$ and $\Hilb_{[n]}$ denotes the subspace with $2\h=S^z=n$, there are two  linearly independent copies
of $\LQGodd$ in $\cent_{H}$  describing
intertwining operators between sectors  with $|j-j'|=1\mod2$. One copy
of (the representation of) $\LQGodd$ in $\cent_{H}$ is generated by $\epr^n\E$ and
$\fpr^m\F$, with $n,m\geq0$, and coincides with the representation
$\repQG(\LQGodd)$. The second copy is generated by  $\epl^n\E$ and
$\fpl^m\F$, with $n,m\geq0$. The two copies are coupled/intersected by the same
$\gl(1|1)$ subalgebra.

\medskip

\section{Conclusion}\label{sec:concl}

The main mathematical result of this paper is the symmetry algebra $\centJTL$ found
in the periodic $\gl(1|1)$ spin-chain -- the centralizer of the
representation of the Jones--Temperley--Lieb algebra $\JTL{N}$. This
symmetry algebra will  be exploited in an analysis~\cite{GRS2} of the
spin-chain as a module over $\JTL{N}$ following earlier results in the boundary case~\cite{ReadSaleur07-2}.

We also discussed in this paper how to proceed from the $\JTL{N}$
generators to get the Virasoro modes in the non-chiral logarithmic
conformal field theory of symplectic fermions:
the combinations $H(n)$ and $P(n)$, introduced in~\eqref{HPn-def},
of the
$\JTL{N}$ generators converge  as $L\to\infty$ to the well-known symplectic fermions representation
of the left and right Virasoro generators
\begin{equation*}
\ffrac{L}{2\pi}H(n) \mapsto 
L_{n}+\bar{L}_{-n}, \qquad \ffrac{L}{2\pi}P(n) \mapsto 
L_{n}-\bar{L}_{-n}.
\end{equation*}

Finally, we showed in Sec.~\ref{sec:symm-continuum} that the scaling limit of the $\JTL{}$ centralizer
$\centJTL$ describes a symmetry of the left-right Virasoro algebra -- that is,  gives
an algebra of intertwining operators respecting the left and right
Virasoro. It is thus reasonable to expect  module structures in the
continuum for the non-chiral Virasoro algebra to be related  to the ones for
$\JTL{N}$: this will be discussed in  our second~\cite{GRS2} and mostly in our third paper \cite{GRS3}.

\medskip

The continuum theory admits a further $\nSU(2)$ symmetry,  which can only lead to
 a refinement of the results inherited from the lattice, since this symmetry is not present in the microscopic model. In fact, if one insists in considering only, in the algebraic approach,  the product 
 $\VirN(2)=\Vir(2)\boxtimes\overline{\Vir}(2)$ as the basic algebra, it is necessary, following our philosophy, to study then the  
  centralizer $\centVir$ of
 $\VirN(2)$ in the local theory. 
 This obviously is not a simple object. It clearly
 contains $\LQG$ (at $\q=i$) generated by the $\nSU(2)$
\begin{equation}\label{eq:Kausch-SU-concl}
Q^a=d^a_{\alpha\beta}\left\{i\phi^\alpha_0\psi^\beta_0+\sum_{m>0}^\infty
\left(\ffrac{\psi^\alpha_{-m}\psi^\beta_m}{m} + \ffrac{\bar{\psi}^\alpha_{-m}\bar{\psi}^\beta_m}{m}\right)\right\}
\end{equation}
  and $\gl(1|1)$ (generated by $\psi_0^1,\psi_0^2$)
but this subalgebra
 does not exhaust the centralizer:
 extending $\gl(1|1)$
we have  the  full scaling limit of the lattice  $\LQGodd$ which we have seen is generated by 
\begin{equation}\label{concl:cent-lim}
  \left[\sum_{m>0}
\Bigl(\ffrac{\sfermp_{m}\sfermp_{-m}}{m} -
\ffrac{\bsfermp_{m}\bsfermp_{-m}}{m}\Bigr)\right]^n\psi_0^2,\qquad
\left[\sum_{m>0}
\Bigl(\ffrac{\sfermm_{m}\sfermm_{-m}}{m} -
\ffrac{\bsfermm_{m}\bsfermm_{-m}}{m}\Bigr)\right]^n\psi_0^1,
\end{equation}
and the  Cartan element, already present in the  $\nSU(2)$~\eqref{eq:Kausch-SU-concl}. On the other hand, it would have been natural  to describe  the
 centralizer  as a quotient of 
 $\LQG\tensor\LQG$ -- the tensor product of the centralizers for (anti)chiral theories $\Vir(2)$ and
 $\overline{\Vir}(2)$. How to do this in practice is not entirely
 clear. We only note that  the centralizer $\centVir$ should in particular contain two
 copies of $\LQGodd$ --- one is the $\JTL{N}$'s centralizer
discussed in Sec.~\ref{sec:scal-lim-centJTL}, and the second is obtained from similar formulas but with
 the opposite sign between the left and right moving components
 in~\eqref{concl:cent-lim} --- coming from each of the chiral halves and coupled
 by the same $\gl(1|1)$ subalgebra.
 However, all these subalgebras still do not   exhaust the centralizer, as can be
 easily seen by commuting them with the $\nSU(2)$.
 
 \medskip
 
We believe  in fact that refining our understanding of $\VirN(2)$ and
$\centVir$ or insisting on the role of Kausch's $\nSU(2)$ is not the way to go. We have strong evidence -- coming
from the study of other models such as those based on $\gl(2|2)$ or
$\gl(2|1)$ that, in fact, the lattice results fully represent the
algebraic structure of the continuum limit. This means that the good
object to consider is not just $\VirN(2)$, but a larger object,
extended by fields mixing the chiral and antichiral sectors, whose
representation theory can be directly inferred from the representation
theory of JTL, and whose centralizer is only $\LQGodd$. This will be
discussed in detail in our third paper~\cite{GRS3}.

\medskip

To conclude, we briefly discuss  the triplet W-algebra~\cite{[K-first],Kausch}. While this algebra does not seem to play an important role in the analysis of models based, {\it e.g.}, on  $\gl(2|2)$ or $\gl(2|1)$, 
it is nevertheless tempting to wonder if, like for the Virasoro algebra, its generators can be simply obtained from lattice considerations. A remark to that effect concerns 
the
permutation $\Pi_{j,j+2}$ of sites at positions $j$, $j+2$. It is easy to
write this operator in terms of fermions
\begin{equation}
\Pi_{j,j+2}=-1+(-1)^j\left(f_j-f_{j+2}\right)\left(f_j^\dagger-f_{j+2}^\dagger\right).
\end{equation}
Consider now 
\begin{equation}
\Pi\equiv \sum_{j=1}^N \Pi_{j,j+2}=2\sum_{\substack{p=\step\\\text{step}=\step}}^{\pi-\step} \left(\chi_p^\dagger\chi_p+\eta_p^\dagger\eta_p\right)(1-\cos 2p).
\end{equation}
In the scaling limit this becomes 
\begin{equation}
\Pi\mapsto \ffrac{1}{L^2}\sum_{m>0} m\left(\psi_{-m}^2\psi_m^1+\psi_{-m}^1\psi_m^2\right)+\left[\psi\to \bar{\psi}\right].
\end{equation}
 We recognize the zero mode of the $W^0+\bar{W}^0$ generator, with
 \begin{equation}
 W^0=\partial\psi^1\psi^2-\psi^1\partial\psi^2,
 \end{equation}
where $\psi^a=i\partial\Phi^a$. It is in fact possible to come up with a full lattice version
 of the triplet W-algebra, with a transparent algebraic interpretation. This   is discussed in a
separate paper~\cite{GST}.

\section*{Acknowledgements}
We are grateful to C.~Candu, I.B.~Frenkel, M.R.~Gaberdiel, J.L.~Jacobsen, 
G.I.~Lehrer, V.~Schomerus, I.Yu.~Tipunin and R.~Vasseur for valuable discussions, and to an anonymous referee for many useful comments. The
work of A.M.G. was supported in part by Marie Curie IIF fellowship, the RFBR grant 10-01-00408,
and the RFBR--CNRS grant 09-01-93105. A.M.G is also grateful to M.R.~Gaberdiel for
kind hospitality in ETH, Z\"urich during  2010, and  grateful to N.~Read for
kind hospitality in Yale University during  2011.  The work of
H.S. was supported by the ANR Projet 2010 Blanc SIMI 4 : DIME.  The work of
N.R. was supported by the NSF grants DMR-0706195 and DMR-1005895.  The authors
are also grateful to the organizers of the ACFTA program at the Institut Henri Poincar\'e in Paris, where this work was finalized.

\section*{Appendix A: The full quantum group $\LQG$ at roots of unity}
\renewcommand\thesection{A}
\renewcommand{\theequation}{A\arabic{equation}}
\setcounter{equation}{0}
We collect here the   expressions for the quantum group $\LQG$ that we
use in the analysis of symmetries of $\gl(1|1)$ spin-chains. We
introduce standard notation for $\q$-numbers
$[n]=\ffrac{\q^n-\q^{-n}}{\q-\q^{-1}}$ and set $[n]!=[1][2]\dots[n]$.

\subsection{Defining relations}
The \textit{full} (or Lusztig) quantum group $\LQG$
with $\q = e^{i\pi/p}$, for $p\geq2$,  is generated by $\E$, $\F$, and
$\K$ satisfying the standard relations for the quantum $s\ell(2)$,
\begin{equation}\label{Uq-com-relations}
  \K\E\K^{-1}=\q^2\E,\quad
  \K\F\K^{-1}=\q^{-2}\F,\quad
  [\E,\F]=\ffrac{\K-\K^{-1}}{\q-\q^{-1}},
\end{equation}
with some constraints,
\begin{equation}\label{root-1-rel}
  \E^{p}=\F^{p}=0,\quad \K^{2p}=\one,
\end{equation}
and 
additionally by the divided powers $\f\sim \F^p/[p]!$ and $\e\sim \E^p/[p]!$, which turn out to satisfy the usual $s\ell(2)$-relations:
\begin{equation}\label{sl2-rel}
  [\h,\e]=\e,\qquad[\h,\f]=-\f,\qquad[\e,\f]=2\h.
\end{equation}
There are also  `mixed' relations~\cite{BFGT}
\begin{gather}
  [\h,\K]=0,\qquad[\E,\e]=0,\qquad[\K,\e]=0,\qquad[\F,\f]=0,\qquad[\K,\f]=0,\label{zero-rel}\\
  [\F,\e]= \ffrac{1}{[p-1]!}\K^p\ffrac{\q \K-\q^{-1} \K^{-1}}{\q-\q^{-1}}\E^{p-1},\qquad
  [\E,\f]=\ffrac{(-1)^{p+1}}{[p-1]!} \F^{p-1}\ffrac{\q \K-\q^{-1} \K^{-1}}{\q-\q^{-1}},
    \label{Ef-rel}\\
  [\h,\E]=\ffrac{1}{2}\E A,\quad[\h,\F]=- \ffrac{1}{2}A\F,\label{hE-hF-rel}
\end{gather}
where 
\begin{equation}\label{A-element}
  A=\,\sum_{s=1}^{p-1}\ffrac{(u_s(\q^{-s-1})-u_s(\q^{s-1}))\K
        +\q^{s-1}u_s(\q^{s-1})-\q^{-s-1}u_s(\q^{-s-1})}{(\q^{s-1}
         -\q^{-s-1})u_s(\q^{-s-1})u_s(\q^{s-1})}\,
        u_s(\K)\idem_s
\end{equation}
with the polynomials $u_s(\K)=\prod_{n=1,\;n\neq s}^{p-1}(\K-\q^{s-1-2n})$, and
$\idem_s$ are some central primitive idempotents~\cite{BFGT}.
The relations~\eqref{Uq-com-relations}-\eqref{A-element} are the defining
relations of the quantum group $\LQG$.

\medskip

The quantum group $\LQG$ has a Hopf-algebra structure
with the comultiplication 
\begin{gather}
  \Delta(\E)=\one\otimes \E + \E\otimes \K,\quad
  \Delta(\F)=\K^{-1}\otimes \F + \F\otimes\one,\quad
  \Delta(\K)=\K\otimes \K,\label{Uq-comult-relations}\\
  \Delta(\e)=\e\tensor\one +\K^p\tensor \e
  +\ffrac{1}{[p-1]!} \sum_{r=1}^{p-1}\ffrac{\q^{r(p-r)}}{[r]}\K^p\E^{p-r}\tensor \E^r
  \K^{-r},\label{e-comult}\\
 \Delta(\f)= \f\tensor \one + \K^p\tensor \f+\ffrac{(-1)^p}{[p-1]!} 
  \sum_{s=1}^{p-1}\ffrac{\q^{-s(p-s)}}{[s]}\K^{p+s}\F^s\tensor \F^{p-s}.\label{f-comult}
\end{gather}
The antipode and counity are not used in the paper but the  reader can
find them, for example, in~\cite{BFGT}.

We can easily write the $(N-1)$-folded coproduct for the capital
generators $\E$ and $\F$,
\begin{equation}\label{N-fold-comult-cap}
\Delta^{N-1}\E =
\sum_{j=1}^{N}\underbrace{\one\tensor\dots\tensor\one}_{j-1}\tensor
\E\tensor \K \tensor \dots \tensor \K,\qquad
\Delta^{N-1}\F =
\sum_{j=1}^{N}\underbrace{\K^{-1}\tensor\dots\tensor \K^{-1}}_{j-1}\tensor
\F\tensor \one \tensor \dots \tensor \one.
\end{equation}

\subsection{Standard spin-chain notations}
We note the Hopf-algebra homomorphism
\begin{equation*}
\E\mapsto S^+ k, \qquad \F\mapsto k^{-1} S^-,\qquad \text{with}\quad k=\sqrt{\K},
\end{equation*}
where we introduced the more usual (in the spin-chain
literature~\cite{PasquierSaleur,DFMC}\footnote{We note that our
  convention for the spin-chain  representation differs from the one
  in~\cite{PasquierSaleur} by the change $\q\to\q^{-1}$.}) quantum group generators
 \begin{equation}
 S^{\pm}=\sum_{1\leq j\leq N} \q^{-\sigma_1^z/2}\otimes\ldots
 \q^{-\sigma_{j-1}^z/2}\otimes\sigma_{j}^{\pm}\otimes\q^{\sigma_{j+1}^z/2}\otimes
 \ldots \otimes \q^{\sigma_{N}^z/2}
 \end{equation}
together with $k=\q^{S^z}$ and the relations
\begin{gather*}
    k S^{\pm}k^{-1}=\q^{\pm 1}S^{\pm},\qquad
   \left[S^{+},S^{-}\right]=\ffrac{k^2-k^{-2}}{\q-\q^{-1}},\\
    \Delta(S^{\pm})=k^{-1}\otimes S^{\pm}+S^{\pm}\otimes k.
\end{gather*}

\subsubsection{The case of XX spin-chains}
For $p=2$ or ``XX spin-chain'' case, the  $(N-1)$-folded coproduct of the renormalized powers $\e$ and $\f$ reads 
\begin{multline}\label{N-fold-comult-ren-e}
\Delta^{N-1}\e =
\sum_{j=1}^{N}\underbrace{\one\tensor\dots\tensor\one}_{j-1}\tensor
\e\tensor \K^2 \tensor \dots \tensor \K^2 +\\
+ \q\sum_{t=0}^{N-2}\sum_{j=1}^{N-1-t} \underbrace{\one\tensor\dots\tensor\one}_{j-1}\tensor
\E\tensor \underbrace{\K \tensor \dots \tensor \K}_{t}\tensor \E\K\tensor \K^2 \tensor \dots \tensor \K^2
\end{multline}
and
\begin{multline}\label{N-fold-comult-ren-f}
\Delta^{N-1}\f =
\sum_{j=1}^{N}\underbrace{\K^2\tensor\dots\tensor \K^2}_{j-1}\tensor
\f\tensor \one \tensor \dots \tensor \one +\\
+ \q^{-1}\sum_{t=0}^{N-2}\sum_{j=1}^{N-1-t}
\underbrace{\K^2\tensor\dots\tensor \K^2}_{j=1}\tensor
\K^{-1}\F\tensor \underbrace{\K^{-1} \tensor \dots \tensor \K^{-1}}_{t}\tensor \F\tensor \one \tensor \dots \tensor \one.
\end{multline}
These renormalized powers can also be expressed in terms of the more usual spin-chain operators, and one finds at $p=2$
\begin{equation*}
\Delta^{N-1}(\e)=\q S^{+(2)}k^{2},\qquad \Delta^{N-1}(\f)=\q^{-1}
k^{-2}S^{-(2)},
\end{equation*}
where $\q=i$ and
 \begin{equation}
 S^{\pm (2)}=\sum_{1\leq j<k\leq N-1} \q^{-\sigma_1^z}\otimes\ldots
 \otimes \q^{-\sigma_{j-1}^z}\otimes\sigma_{j}^{\pm}
 \otimes1\otimes\ldots\otimes 1\otimes \sigma_k^{\pm} \otimes
 \q^{\sigma_{k+1}^z}\otimes \ldots \otimes \q^{\sigma_{N}^z}.
 \end{equation}

\section*{Appendix B: A proof for the centralizer of $\JTL{N}$}
\renewcommand\thesection{B}
\renewcommand{\theequation}{B\arabic{equation}}
\setcounter{equation}{0}

Our proof of Thm.~\ref{Thm:centr-JTL-main} consists in the following three lemmas.  First,
in Lem.~\ref{lem:endomor-TL}, we describe the two-parameter family of
vector spaces $\VEnd_{k,t}$ spanned by homomorphisms (respecting the
open Temperley--Lieb algebra $\TL{N}$ generated by $e_j$ with $1\leq
j\leq N-1$) between any two sectors $\Hilb_{(k)}$ and $\Hilb_{(k')}$ for
$0\leq k\leq N$ and $k-k'=1\mod 2$; we denote by $\Hilb_{(k)}$  the sector\footnote{We note the
  subspace $\Hilb_{(k)}$ coincides with $\Hilb_{[\frac{N}{2}-k]}$, where $\Hilb_{[n]}$ denotes the subspace with $2\h=S^z=n$.}
with $k$~an\-ti\-fer\-mions $\ferm_{p_j}$. Then, in
Lem.~\ref{lem:endo-perTL-count}, we compute commutators between $e_N$
and an intertwining operator from $\VEnd_{k,t}$ and show that all homomorphisms
(between $\Hilb_{(k)}$ and $\Hilb_{(k')}$) respecting the periodic
Temperley--Lieb\footnote{That  is, the algebra
  generated by the $e_j$ with $1\leq j\leq N$, {\it i.e.}, without the
  translation generator $u^2$, see Sec.~\ref{sec:super-spin-ch-def}.} algebra $\TL{N}^a$ are exhausted by elements from
$\LQGodd$. In Lem.~\ref{lem:comm-perTL-final}, we state that
all homomorphisms  between $\Hilb_{(k)}$ and
$\Hilb_{(k')}$ (as modules over $\TL{N}^a$) for $k-k'=0\mod 2$ are also given by $\LQGodd$, together with the two operators $\e^L$ and $\f^L$ mixing
the two $\TL{N}^a$-invariants on the opposite ends of the spin-chain. 
We finally state an isomorphism between the centralizers for
(the $\gl(1|1)$ representations of) $\TL{N}^a$ and $\JTL{N}$.

In what follows,  we omit the notation for the spin-chain representation
$\repQG$ of the quantum group  for brevity and  simply write $\F$ or $\E$ instead of
$\repQG(\F)$ or $\repQG(\E)$. We do the same for the representation
$\repgl$ of generators of $\JTL{N}$.

\begin{Lemma}\label{lem:endomor-TL}
The vector space $\VEnd_{k,t}$ of homomorphisms respecting the $\TL{N}$-action,
\begin{equation*}
\VEnd_{k,t}=\Hom_{\TL{N}}(\Hilb_{(k)},\Hilb_{(k-2t-1)}),\qquad 0\leq
k\leq N,\quad
\Big\lceil\ffrac{k-N-1}{2}\Big\rceil\leq t\leq \Big\lfloor\ffrac{k-1}{2}\Big\rfloor,
\end{equation*}
has the dimension and a basis listed below.
\begin{enumerate}
\item For $0\leq t\leq \big\lfloor\frac{k-1}{2}\big\rfloor$, we have
\begin{itemize}
\item for $k\leq\frac{N}{2}$,
and also for $k>\frac{N}{2}$ and $k-\frac{N}{2}\leq t\leq
\big\lfloor\ffrac{k-1}{2}\big\rfloor$,
\begin{align}
\VEnd_{k,t}
&= \Big\langle \f^{n-t}\e^n\E,\, \F\f^{l-t}\e^{l+1};
\, t\leq n\leq \Big\lfloor\ffrac{k-1}{2}\Big\rfloor,\, t\leq l\leq
\Big\lfloor\ffrac{k-2}{2}\Big\rfloor\Big\rangle,\label{eq:Ekt-span-plus-1}\\
\dim\VEnd_{k,t}&=k-2t,\notag
\end{align} 
\item for $k>\frac{N}{2}$ and $0\leq t\leq k-\frac{N}{2}-1$,
\begin{align}
\VEnd_{k,t}
&= \Big\langle \e^t\E,\, \f^{n-t}\e^n\E,\, \F\f^{l-t}\e^{l+1};
\, k-\ffrac{N}{2}\leq n\leq \Big\lfloor\ffrac{k-1}{2}\Big\rfloor,\,
k-\ffrac{N}{2} \leq l\leq
\Big\lfloor\ffrac{k-2}{2}\Big\rfloor\Big\rangle,\label{eq:Ekt-span-plus-2}\\
\dim\VEnd_{k,t}&=N-k+1,\notag
\end{align} 
\end{itemize}

\item For $\big\lceil\frac{k-N-1}{2}\big\rceil\leq t\leq -1$, we
have
\begin{itemize}
\item for $k\geq\frac{N}{2}$,
and also for $k<\frac{N}{2}$ and $\big\lceil\frac{k-N-1}{2}\big\rceil\leq t\leq
k-\frac{N}{2}$,
\begin{align*}
\VEnd_{k,t}
&= \Big\langle \e^{n+t+1}\f^n\F,\, \E\e^{l+t+1}\f^{l+1};
\, -t-1\leq n\leq \Big\lfloor\ffrac{N-k-1}{2}\Big\rfloor,\, -t-1\leq l\leq
\Big\lfloor\ffrac{N-k-2}{2}\Big\rfloor\Big\rangle,
\\
\dim\VEnd_{k,t}&=N-k+2(t-1),\notag
\end{align*} 
\item for $k<\frac{N}{2}$ and $k-\frac{N}{2}+1\leq t\leq -1$,
\begin{align}
\VEnd_{k,t}
&= \Big\langle \f^{-t-1}\F,\, \e^{n+t+1}\f^n\F,\, \E\e^{l+t+1}\f^{l+1};
\, \ffrac{N}{2}-k\leq n\leq \Big\lfloor\ffrac{N-k-1}{2}\Big\rfloor,\label{eq:Ekt-span-minus-2}\\
&\qquad\qquad\qquad\qquad\qquad\qquad\qquad\qquad\ffrac{N}{2}-k \leq l\leq
\Big\lfloor\ffrac{N-k-2}{2}\Big\rfloor\Big\rangle,\notag\\
\dim\VEnd_{k,t}&=k,\notag
\end{align} 
\end{itemize}
\end{enumerate}
where we suppose that each basis element is multiplied by an
appropriate projector on the sector $\Hilb_{(k)}$ -- a polynomial in the
Cartan element $\h$.
\end{Lemma}
\begin{proof}
The idea of the proof is to compute  dimensions of the spaces $\VEnd_{k,t}$ of
homomorphisms using explicit decompositions over the two commuting algebras and
then to check that the images
(in $\Hilb_{(k-2t-1)}$)
of the  basis elements proposed in the lemma are non-isomorphic, so they are indeed linearly independent.

We recall the decomposition of the tensor-product space $\Hilb_{N}$ over the two
commuting algebras $\TL{N}$ and $\LQG$ (centralizing each other) in
the open case~\cite{ReadSaleur07-2},
\begin{equation}\label{decomp-LQG-TL}
\Hilb_{N}|_{\rule{0pt}{7.5pt}%
\LQG} =  \bigoplus_{j=1}^{L}\IrrTL{j}\boxtimes \PP_{1,j}\,, \qquad 
\Hilb_{N}|_{\rule{0pt}{7.5pt}%
\TL{N}} =  \bigoplus_{j=1}^{L} \PrTL{j}\boxtimes \XX_{1,j} \oplus \StTL{L}\boxtimes \XX_{1,L+1}, 
\end{equation}
 with multiplicities $d^0_j=\sum_{i=j}^L(-1)^{j-i}\left(\binom{N}{L+i}
 - \binom{N}{L+i+1}\right)$ given by dimensions of irreducibles over
 $\TL{N}$. We use the notations $\PrTL{j}$ and $\StTL{j}$ for
 projective and standard $\TL{N}$-modules, respectively.  The standard
 module $\StTL{L}$ is the trivial representation denoted also by
 $(1)$; the standard module $\StTL{j}$, with $1 \leq j < L$, has the
 dimension $\binom{N}{L+j} - \binom{N}{L+j+1}$ and is indecomposable,
 with the structure of subquotients $\StTL{j}:
 \IrrTL{j}\rightarrow\IrrTL{j+1}$, where by $ \IrrTL{j}$ we denote irreducible TL modules. The projectives $\PrTL{j}$ are
 described by the diagram
 $\StTL{j}\rightarrow\StTL{j-1}$ or with  simple subquotients as
\begin{equation}\label{schem-projTL}
     \xymatrix@C=2pt@R=24pt@M=1pt@W=2pt{&&\\&\PrTL{j}\quad=\quad&\\&&}
\xymatrix@C=4pt@R=18pt@M=2pt@W=2pt{
           &\IrrTL{j} \ar[dl] \ar[dr]&\\
 	  \IrrTL{j-1} \ar[dr]
 	    &&\IrrTL{j+1}\ar[dl]\\
        &\IrrTL{j}&
    } 
\end{equation}
 where we exclude subquotients $\IrrTL{j>L}$ from the diagram, 
 see also more details in~\cite{ReadSaleur07-2} including the bimodule structure. We note that these modules are self-contragredient\footnote{Recall that the  (left) $A$-module $V^*$ contragredient to a left $A$-module $V$ is  the vector space of linear functions $V\to\oC$ with the left action of the algebra $A$ given by $a\, f(v) = f(a^{\dagger}v)$  for any $v\in V$ and $f\in V^*$. We use the anti-involution 
$\cdot^{\dagger}: (e_i)^{\dagger}=e_{N-i}$ on the TL algebra.
The contragredient module is then described by a diagram where all arrows are inverse with respect to the diagram of the initial module.}, {\it i.e.}, $\PrTL{j}\cong\PrTL{j}^{*}$. On the quantum-group
 side, the $\LQG$-action on the $j$-dimensional irreducible modules
 $\XX_{1,j}$ is defined in~\eqref{eq:sl-irrep} and the action on the
 projective modules $\PP_{1,j}$ is defined
 in~\eqref{schem-proj}-\eqref{proj-mod-QGbase}.  We note also the
 only non-trivial $\Hom$ spaces for a pair of projectives over
 $\TL{N}$ are
 $\HommTL(\PrTL{j},\PrTL{j})\cong\oC^2$ and $\HommTL(\PrTL{j},\PrTL{j\pm1})\cong\oC$.

Using the decomposition~\eqref{decomp-LQG-TL} over $\TL{N}$ restricted
to sectors with $S^z=k$ and $S^z=k-2t-1$ as well as the $\HommTL$ spaces
for a pair of projective  $\TL{N}$-modules described just above, we  easily compute
dimensions of the spaces $\VEnd_{k,t}$ for all cases described in the
lemma.

Next,  in order to describe images of intertwining operators $\e^m\f^n\E$
 and $\e^m\f^n\F$ in each subspace $\Hilb_{(k)}$ we introduce ``zig-zag'' type
$\TL{N}$-modules in Fig.~\ref{TL-Mmod}.
\begin{figure}\centering
\begin{align*}
\xymatrix@R=15pt@C=3pt@W=2pt@M=1pt
{
&\mbox{}\\
&\boldsymbol{\MJTL{n-1}\;:}
}&
\xymatrix@R=26pt@C=6pt@W=4pt@M=4pt
{
&{(\TLX_n)}\ar[dr]\ar[dl]&&
{(\TLX_{n+2})}\ar[dr]\ar[dl]&&
{(\TLX_{n+4})}\ar[dr]\ar[dl]&
\dots&
{(\TLX_{L-1})}\ar[dr]\ar[dl]&
\\
 (\TLX_{n-1})&
 &(\TLX_{n+1})&
   &(\TLX_{n+3})&
   &\mbox{}\quad\dots\quad\mbox{}&
   &(\TLX_{L})&
     }\\
&\mbox{}\\
\xymatrix@R=15pt@C=3pt@W=2pt@M=1pt
{
&\mbox{}\\
&\boldsymbol{\NJTL{n-1}\;:}
}&
\xymatrix@R=26pt@C=22pt@W=4pt@M=4pt
{
&{(\TLX_n)}\ar[dr]\ar[d]&
{(\TLX_{n+2})}\ar[dr]\ar[d]&
{(\TLX_{n+4})}\ar[dr]\ar[d]&
\dots\ar[dr]&
{(\TLX_{L})}\ar[d]&
\\
& (\TLX_{n-1})&
 (\TLX_{n+1})&
  (\TLX_{n+3})&
  \mbox{}\quad\dots\quad\mbox{}&
   (\TLX_{L-1})&
     }
\end{align*}
\caption{For even $L$, the indecomposable ``zig-zag'' $\TL{N}$-modules
  $\MJTL{n-1}$ at the top, with odd $n$, 
  and  $\NJTL{n-1}$ at the bottom, with even $n$,
   $k=(L-\bar{n})/2$, and $\bar{n} = n - (n\,\textrm{mod}\,2)$.}
    \label{TL-Mmod}
    \end{figure}
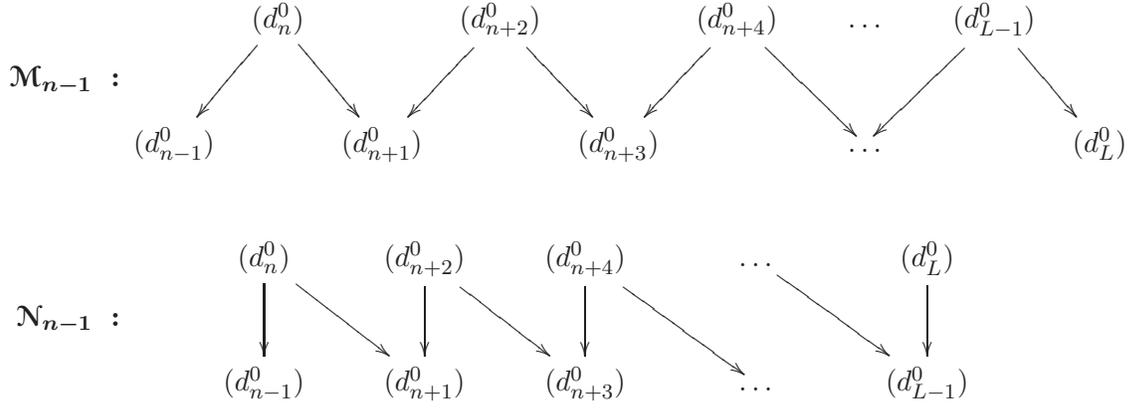
 These are obtained as kernels of $\F$ or $\E$ in the following way.
We note  that the spin-chain $\Hilb_N=\oplus_{j=-L}^{j=L}\Hilb_{[j]}$
graded by $S^z$ defines two  long exact
sequences with the differentials $\F$ and $\E$
(we recall that $\F^2=\E^2=0$) with
$\E:\Hilb_{[n]}\to\Hilb_{[n+1]}$ and $\F:\Hilb_{[n]}\to\Hilb_{[n-1]}$.  The images and kernels of these
differentials are $\TL{N}$-modules: 
for any
$j>0$, using again the decompositions~\eqref{decomp-LQG-TL} restricted
to the supbspace with $S^z=j$ and the $\LQG$-action from
App.~\Approjmodbase, we obtain the short exact sequences of
$\TL{N}$-modules
\begin{align*}
\mbox{}\qquad\qquad
0\;\to\;\NJTL{1}\;\to&\;\Hilb_{[0]}\;\to\;\NJTL{1}^*\;\to\;0,&\\
0\;\to\;\MJTL{j+1}\;\to&\;\Hilb_{[j]}\;\to\;\NJTL{j}\;\to\;0,
\qquad  j - \,\text{odd},&\\
0\;\to\;\NJTL{j+1}\;\to&\;\Hilb_{[j]}\;\to\;\MJTL{j}\;\to\;0,
\qquad  j - \,\text{even},&
\end{align*}
where we define the submodules $\MJTL{j+1}$ and
$\NJTL{j+1}$ as the kernels of the
quantum-group generator $\F$ on $\Hilb_{[j]}$,  for odd and even $j$, respectively. 
Equivalently, they are defined as  the kernels
of $\E$ but for $j<0$.

We note next that the modules $\MJTL{j}$ and
$\NJTL{j}$ have filtrations by submodules $\MJTL{j'}$ and
$\NJTL{j'}$ with appropriate $j'>j$.  
We then use  the bimodule structure on $\Hilb_N$ described in~\cite{ReadSaleur07-2}
together with the explicit action of
$\LQG$ generators $\e$ and $\f$  given also in App.~\Approjmodbase~in order to compute images of the intertwining operators $\e^m\f^n\E$ and $\e^m\f^n\F$ (these images are 
identified with  terms of the
filtrations in the submodules $\MJTL{j}$ and $\NJTL{j}$.)
By straightforward calculations we check that the images of  intertwining operators
proposed in the lemma for all possible pairs $(k,t)$ are given by non-isomorphic
``zig-zag'' type $\TL{N}$-modules (together with their duals), defined
just above and described in Fig.~\ref{TL-Mmod}. This finishes our proof.
\end{proof}


\begin{Lemma}\label{lem:endo-perTL-count}
All homomorphisms between $\Hilb_{(k)}$ and $\Hilb_{(k')}$, for $0\leq
k\leq N$ and $k-k'=1\mod 2$, respecting the periodic Temperley--Lieb
algebra $\TL{N}^a$ are given by action of elements from $\LQGodd$.
\end{Lemma}
\begin{proof}			
We use the fermionic expressions~\eqref{QG-ferm-2}-\eqref{QG-ferm-3}
for the generators of $\LQG$ and the explicit expressions for the
commutators~\eqref{eq:comm-TL-f} and~\eqref{eq:comm-TL-e} to find, for
$n\leq N/2-1$,
\begin{align*}
[e_N,\f^n] &= 2^{n}i\sum_{p_1,\dots,p_n\ne\frac{3\pi}{2}}(e^{ip_{n}}-i) 
\prod_{j=1}^{n-1}f(p_j)\ferm_{p_j}\ferm_{\pi-p_j}\ferm_{p_n}\ferm_{3\pi/2},\\
[e_N,\e^n] &=
2^{n}i\sum_{p_1,\dots,p_n\ne\frac{\pi}{2}}(e^{-ip_{n}}-i)\prod_{j=1}^{n-1}
g({p_j})\fermd_{p_j}\fermd_{\pi-p_j}\fermd_{p_n}\fermd_{\pi/2},
\end{align*}
where 
\begin{equation*}
f(p_j) =
\q\frac{e^{i(\frac{3\pi}{2}-p_j)}}{e^{i(\frac{3\pi}{2}-p_j)}-1},\qquad
g(p_j) =
-\q\frac{e^{i(\frac{\pi}{2}+p_j)}}{e^{i(\frac{\pi}{2}+p_j)}+1}.
\end{equation*}
Simplifying, we get
\begin{align}
[e_N,\f^n] &=
(-2)^n i(n-1)!\sum_{\substack{\frac{\pi}{2}=p_1<p_2<\dots<p_{n-1}\\\text{step}=\step}}^{\frac{3\pi}{2}-\step}
\sum_{p_n}(i-e^{ip_{n}}) 
\prod_{j=1}^{n-1}\cot\ffrac{1}{2}(\ffrac{\pi}{2}+p_j)\,\ferm_{p_j}\ferm_{\pi-p_j}\ferm_{p_n}\ferm_{3\pi/2}\ne 0,\label{eq:comm-TL-fn}\\
[e_N,\e^n] &= 2^n i(n-1)!\sum_{\substack{\frac{\pi}{2}+\step=p_1<p_2<\dots<p_{n-1}\\\text{step}=\step}}^{\frac{3\pi}{2}}
\sum_{p_n}(e^{-ip_{n}}-i) \prod_{j=1}^{n-1}\tan\ffrac{1}{2}(\ffrac{\pi}{2}+p_j)\,\fermd_{p_j}\fermd_{\pi-p_j}\fermd_{p_n}\fermd_{\pi/2}\ne 0,\label{eq:comm-TL-en}
\end{align}
where we introduce $\step=\frac{2\pi}{N}$ and all the non-zero
summands are linearly independent.

Then, using~\eqref{QG-ferm-1}-\eqref{QG-ferm-3}, we obtain
\begin{align*}
\e^m \E &= m!\sqrt{N}\!\!\!\sum_{\substack{\frac{\pi}{2}+\step=p_1<p_2<\dots<p_{m}\\\text{step}=\step}}^{\frac{3\pi}{2}}
\prod_{j=1}^{m}\tan\ffrac{1}{2}(\ffrac{\pi}{2}+p_j)\,\fermd_{p_j}\,\fermd_{\pi-p_j}\fermd_{\pi/2}K,\\
\F\f^n &=(-1)^{n-1} n! i \sqrt{N}\!\!\!\sum_{\substack{\frac{\pi}{2}=p_1<p_2<\dots<p_{n}\\\text{step}=\step}}^{\frac{3\pi}{2}-\step}
\prod_{j=1}^{n}\cot\ffrac{1}{2}(\ffrac{\pi}{2}+p_j)\,\ferm_{p_j}\,\ferm_{\pi-p_j}\ferm_{3\pi/2}
\end{align*}
and, using~\eqref{eq:comm-TL-efn} and~\eqref{eq:comm-TL-fn},
\begin{multline}\label{eq:comm-eN-fnemE}
[e_N,\f^n\e^m\E] = [e_N,\f^n]\e^m \E =\\
=(-2)^n i(n-1)!\, m!\sqrt{N}\!\!\!\!\!\!\sum_{\substack{\frac{\pi}{2}=p_1<p_2<\dots<p_{n-1}\\\text{step}=\step}}^{\frac{3\pi}{2}-\step}
\sum_{p_n}
\sum_{\substack{\frac{\pi}{2}+\step=p'_1<p'_2<\dots<p'_{m}\\\text{step}=\step}}^{\frac{3\pi}{2}}
\!\!\!\!\!\!(i-e^{ip_{n}}) 
\prod_{j=1}^{n-1}\cot\ffrac{1}{2}(\ffrac{\pi}{2}+p_j)\,\ferm_{p_j}\ferm_{\pi-p_j}\ferm_{p_n}\ferm_{3\pi/2}\\
\times
\prod_{l=1}^{m}\tan\ffrac{1}{2}(\ffrac{\pi}{2}+p'_l)\,\fermd_{p'_l}\,\fermd_{\pi-p'_l}\fermd_{\pi/2}
K, \qquad n\geq 1,\, m\geq 0,
\end{multline}
and, similarly,
\begin{multline}\label{eq:comm-eN-Ffnem}
[e_N,\F\f^n\e^m] = \F\f^n[e_N,\e^m] =\\
=2^m(-1)^{n} (m-1)!\, n! \sqrt{N}\!\!\!\sum_{\substack{\frac{\pi}{2}=p_1<p_2<\dots<p_{n}\\\text{step}=\step}}^{\frac{3\pi}{2}-\step}
\sum_{p'_m}
\sum_{\substack{\frac{\pi}{2}+\step=p'_1<p'_2<\dots<p'_{m-1}\\\text{step}=\step}}^{\frac{3\pi}{2}}
\!\!\!\!\!\!(e^{-ip'_{m}}-i)
\prod_{j=1}^{n}\cot\ffrac{1}{2}(\ffrac{\pi}{2}+p_j)\,\ferm_{p_j}\,\ferm_{\pi-p_j}\ferm_{3\pi/2}\\
\times
\prod_{l=1}^{m-1}\tan\ffrac{1}{2}(\ffrac{\pi}{2}+p'_l)\,\fermd_{p'_l}\fermd_{\pi-p'_l}\fermd_{p'_m}\fermd_{\pi/2},
\qquad m\geq 1,\, n\geq 0,
\end{multline}
 where all the summands are linearly independent.

Next, we restrict the action on a sector with 
$k$ antifermions $\ferm_{p_j}$, $0\leq k\leq N$,
\begin{equation*}
\Hilb_{[\frac{N}{2}-k]}\equiv\Hilb_{(k)}=
\Big\langle\prod_{j=1}^{k}\ferm_{p_j}\vectv{\uparrow\dots\uparrow},
\; p_1>p_2>\dots>p_k\Big\rangle,
\end{equation*}
where the momenta $p_j$ belong to the set~\eqref{momenta-set},
and compute commutation relations between $e_N$ and an operator from 
the vector space
$\VEnd_{k,t}=\Hom_{\TL{N}}(\Hilb_{(k)},\Hilb_{(k-2t-1)})$ described
in Lem.~\ref{lem:endomor-TL}. 
We first note that the intersection of $\VEnd_{k,t}$ with $\LQGodd$ is spanned by
the operators $\e^t \E$, for $t\geq 0$, and by $f^{-t-1}F$, for $t\leq -1$.
Then, we show in three steps that any non-zero linear combination of all the other operators
is not contained in the space
$\Hom_{\TL{N}^a}(\Hilb_{(k)},\Hilb_{(k-2t-1)})$ of homomorphisms
respecting the periodic Temperley--Lieb algebra $\TL{N}^a$. We begin
with consideration of the case $\textit{1.}$ in Lem.~\ref{lem:endomor-TL}.
\begin{enumerate}
\item
Using~\eqref{eq:comm-eN-fnemE}, we calculate the action of the
commutator $[e_N,\f^m\e^n\E]$ on a vector
$v_k(n)=\ferm_{p''_1}\ferm_{p''_2}\dots\\\dots\ferm_{p''_k}\vectv{\uparrow\dots\uparrow}$
with the specifically chosen momenta
\begin{equation}\label{eq:choosn-mom}
p''_j=\ffrac{\pi}{2}+(n-j+1)\step, \qquad 1\leq j\leq k,\quad 1\leq k\leq N.
\end{equation}
Here, we set $m=n-t$ and the power $n$ is running the values $t+1,\dots,
\big\lfloor\frac{k-1}{2}\big\rfloor$ in the case corrseponding
to~\eqref{eq:Ekt-span-plus-1}, and for the
case~\eqref{eq:Ekt-span-plus-2} -- the values $k-\frac{N}{2},\dots,
\big\lfloor\frac{k-1}{2}\big\rfloor$.
\begin{multline*}
[e_N,\f^m\e^n\E]v_k(n) = [e_N,\f^m\e^n\E]\,\ferm_{\frac{\pi}{2}+n\step}\,\ferm_{\frac{\pi}{2}+(n-1)\step}
\dots\ferm_{\frac{\pi}{2}-n\step}\dots\ferm_{\frac{\pi}{2}+(n-k+1)\step}\vectv{\uparrow\dots\uparrow}\\
=a(n)\sum_{\substack{\frac{\pi}{2}+\step=p_1<p_2<\dots<p_{m-1}\\\text{step}=\step}}^{\frac{3\pi}{2}-\step}
\sum_{p_m} (i-e^{ip_{m}}) 
\prod_{j=1}^{m-1}\cot\ffrac{1}{2}(\ffrac{\pi}{2}+p_j)\,\ferm_{p_j}\ferm_{\pi-p_j}\ferm_{p_m}\ferm_{\frac{3\pi}{2}}\,
\times\\
\times\!\!\!\!\!\!\!
\sum_{\substack{\frac{\pi}{2}+\step=p'_1<p'_2<\dots<p'_{n}\\\text{step}=\step}}^{\frac{3\pi}{2}-\step}
\prod_{l=1}^{n}\tan\ffrac{1}{2}(\ffrac{\pi}{2}+p'_l)\,\fermd_{p'_l}\,\fermd_{\pi-p'_l}\fermd_{\frac{\pi}{2}}
\ferm_{\frac{\pi}{2}+n\step}\,\ferm_{\frac{\pi}{2}+(n-1)\step}\dots\ferm_{\frac{\pi}{2}+(n-k+1)\step}\vectv{\uparrow\dots\uparrow},
\end{multline*}
where $a(n) = (-1)^{m-k+N/2}\, i\, 2^m n! (m-1)!\sqrt{N}$. In
consequence of the condition $\frac{\pi}{2}+(n-k+1)\step\geq
-\frac{\pi}{2}+\step$ or, equivalently, $k-n\leq N/2$, the state
$v_k(n)$ contains precisely $n$ pairs of momenta $(p_l,\pi-p_l)$,
where $\frac{\pi}{2}+n\step\leq p_l\leq
\frac{\pi}{2}+\step$. Therefore, the sum over $p'_l$ (on the third
line) contains only one non-zero term corresponding to
$p'_l=\frac{\pi}{2}+l\step$, $1\leq l\leq n$:
\begin{multline*}
\prod_{l=1}^{n}\tan\ffrac{1}{2}(\ffrac{\pi}{2}+l\step)\,\fermd_{\frac{\pi}{2}+\step}\,\fermd_{\frac{\pi}{2}-\step}
\,\fermd_{\frac{\pi}{2}+2\step}\,\fermd_{\frac{\pi}{2}-2\step}\dots
\fermd_{\frac{\pi}{2}+n\step}\,\fermd_{\frac{\pi}{2}-n\step}\,\fermd_{\frac{\pi}{2}}\times\\
\times\underbrace{\ferm_{\frac{\pi}{2}+n\step}\,\ferm_{\frac{\pi}{2}+(n-1)\step}\dots
\ferm_{\frac{\pi}{2}+\step}\,\ferm_{\frac{\pi}{2}}\ferm_{\frac{\pi}{2}-\step}\dots
\ferm_{\frac{\pi}{2}-n\step}}\,\ferm_{\frac{\pi}{2}-(n+1)\step}\dots\ferm_{\frac{\pi}{2}+(n-k+1)\step}\vectv{\uparrow\dots\uparrow},
\end{multline*}
where the under-braced term is annihilated by the corresponding
$\fermd$-operators. We thus obtain
\begin{multline*}\label{eq:comm-eN-fnenE-k}
[e_N,\f^m\e^n\E]v_k(n)= a(n)\!\!\! \sum_{\substack{\frac{\pi}{2}+\step=p_1<p_2<\dots<p_{m-1}\\\text{step}=\step}}^{\frac{3\pi}{2}-\step}
\sum_{p_m} (i-e^{ip_{m}}) 
\prod_{l=1}^{n}\tan\ffrac{1}{2}(\pi+l\step)\prod_{j=1}^{m-1}\cot\ffrac{1}{2}(\ffrac{\pi}{2}+p_j)\\
\times\ferm_{p_j}\ferm_{\pi-p_j}\ferm_{p_m}\ferm_{\frac{3\pi}{2}}
\ferm_{\frac{\pi}{2}-(n+1)\step}\dots\ferm_{\frac{\pi}{2}+(n-k+1)\step}\vectv{\uparrow\dots\uparrow}\ne
0,
\end{multline*}
where the sum contains non-zero terms in consequence of the inequality
$m<n$ and all the non-zero terms are linearly independent.

\item
Using~\eqref{eq:comm-eN-Ffnem}, we next calculate similarly the action
of the commutator $[e_N,F\f^m\e^{n+1}]$, where $m=n-t$ and $t\leq
n\leq \big\lfloor\frac{k-2}{2}\big\rfloor$ for the
case~\eqref{eq:Ekt-span-plus-1} and $k-\frac{N}{2}\leq n\leq
\big\lfloor\frac{k-2}{2}\big\rfloor$ in the case~\eqref{eq:Ekt-span-plus-2}, on the same vector
$v_k(n)=\ferm_{p''_1}\ferm_{p''_2}\dots\ferm_{p''_k}\vectv{\uparrow\dots\uparrow}$
with the momenta~\eqref{eq:choosn-mom}.
\begin{multline*}
[e_N,\F\f^m\e^{n+1}]v_k(n)= b(n)\!\! \sum_{\substack{\frac{\pi}{2}+\step=p_1<p_2<\dots<p_{m}\\\text{step}=\step}}^{\frac{3\pi}{2}-\step}
\!\!\sum_{r=1}^{k-2n-1}(-1)^{r-1} 
\prod_{j=1}^{m}\cot\ffrac{1}{2}(\ffrac{\pi}{2}+p_j)\prod_{l=1}^{n}\tan\ffrac{1}{2}(\pi+l\step)\\
\times (e^{-i(\frac{\pi}{2}-(n+r)\step)}-i)\, \ferm_{p_j}\ferm_{\pi-p_j}\ferm_{\frac{3\pi}{2}}
\ferm_{\frac{\pi}{2}-(n+1)\step}\dots\widehat{\ferm}_{\frac{\pi}{2}-(n+r)\step}
\dots\ferm_{\frac{\pi}{2}+(n-k+1)\step}\vectv{\uparrow\dots\uparrow}\ne0,
\end{multline*}
where $b(n) = (-1)^m 2^{n+1} m! n! \sqrt{N}$ and the notation
$\widehat{\ferm}$ means the absence of the corresponding term.

\item
We then note that, for $t\leq r\leq n-1$ in the
case corresponding to~\eqref{eq:Ekt-span-plus-1} and $k-\frac{N}{2}\leq r\leq n-1$ in
the case~\eqref{eq:Ekt-span-plus-2},
\begin{equation*}		
[e_N,\alpha_n\f^{n-t}\e^{n}\E + \beta_n\F\f^{n-t}\e^{n+1}]v_k(r) =0, \quad
\;\; \alpha_n,\beta_n\in\oC,
\end{equation*}
and 
\begin{equation*}
[e_N,\alpha_n\f^{n-t}\e^{n}\E + \beta_n\F\f^{n-t}\e^{n+1}]v_k(n) \ne
0, 
\end{equation*}
for any non-zero complex numbers $\alpha_n$ and $\beta_n$. Therefore,
we can prove by induction with respect to $n$
that the only operator
from $\VEnd_{k,t}$ which commutes with $e_N$ is $\e^t\E\in\LQGodd$.
\end{enumerate}
Similar analysis can be carried out for the case $\textit{2.}$ of
Lem.~\ref{lem:endomor-TL}. More convenient basis to express $v_k(n)$
for this case is spanned by
$\prod_{j=1}^{k}\fermd_{p_j}\vectv{\downarrow\dots\downarrow}$.
\end{proof}

\begin{Lemma}\label{lem:comm-perTL-final}
The vector space $\Hom_{\TL{N}^a}(\Hilb_{(k)},\Hilb_{(k-2t)})$ of
homomorphisms between $\Hilb_{(k)}$ and $\Hilb_{(k')}$, for
$k-k'=0\mod 2$, respecting $\TL{N}^a$ is spanned by elements from
$\LQGodd$ and operators $\f^L$, $\e^L$.
\end{Lemma}
\begin{proof}
First, using the
fermionic expression~\eqref{eq:comm-TL-fn} and~\eqref{eq:comm-TL-en}
for $[e_N,\f^n]$ and $[e_N,\e^n]$, we
conclude that the only power of $\f$ and $\e$ that commutes with $e_N$ is
$n=N/2=L$.

Second, we  note that (for $t\geq 0$) the vector space of operators
intertwining the $\TL{N}$-action has the basis
\begin{equation*}
\tilde{\VEnd}_{k,t}\equiv\Hom_{\TL{N}}(\Hilb_{(k)},\Hilb_{(k-2t)})=\Big\langle
\proj_k(\h)\delta_{t,0},\, \f^{n-t}\e^n\F^{\nu}\E^{\nu};
\, \nu\in\{0,1\},\; t\leq n\leq
\Big\lfloor\ffrac{k-1}{2}\Big\rfloor\Big\rangle,
\end{equation*}
where $\delta$ is the Kronecker symbol and we introduce
polynomials $\proj_k(\h)$ in $\h$ projecting the tensor-product space
$\Hilb_{N}$ onto the subspace $\Hilb_{(k)}$,
\begin{equation}\label{proj-def}
\proj_k(\h) = \prod_{j=-N/2;\,j\ne L-k}^{j=N/2}(2\h-j).
\end{equation}
The intersection of $\tilde{\VEnd}_{k,t}$ with $\LQGodd$ is spanned by
the operators  $\f^{n-t}\e^n\F\E$ (and $\proj_k(\h)$, for $t=0$).
Assume there is a non-zero operator $\mathcal{O}$ from
$\tilde{\VEnd}_{k,t}$ which is not presented by an element from
$\LQGodd$ but commutes with $e_N$ and consider the product of the
operator $\mathcal{O}$ with $\E$. Obviously, the operator
$\mathcal{O}\E$ is non-zero because by our assumption $\mathcal{O}$ is
not in $\LQGodd$ and therefore it is a linear
combination of projectors on a $\TL{N}$ direct summand times $\e^t$
which does not belong to the kernel of $\E$. Then, the assumption
$[\mathcal{O},e_N]=0$ with the fact $[\E,e_N]=0$ imply that the
non-zero homomorphism
$\mathcal{O}\E\in\VEnd_{k,t}=\Hom_{\TL{N}}(\Hilb_{(k)},\Hilb_{(k-2t-1)})$
which is not represented by an element from $\LQGodd$
commutes with $e_N$. We thus get a contradiction to
Lem.~\ref{lem:endo-perTL-count}. This finishes our proof for $t\geq
0$. The case $t<0$ is considered in the same manner.

 We finally note that operators $\EE n\f^L$ and $\FF n\e^L$
belong to $\repQG\bigl(\LQGodd\bigr)$. Therefore, the multiplication of $\f^L$
or $\e^L$ with `off-diagonal' elements from $\LQGodd$ does not give any new
operatots centralizing $\TL{N}^a$.
\end{proof}

Combining all the three Lemmas we prove the following result.
\begin{Cor}\label{cor:TLa-centralizer}
The centralizer of the
periodic Temperley--Lieb algebra $\repgl(\TL{N}^a)$ on the spin-chain
is isomorphic to the subalgebra in $\repQG\bigl(\LQG\bigr)$ generated
by $\LQGodd$ and $\f^L$, and $\e^L$.
\end{Cor}

A final comment is in order however: the algebra $\JTL{N}$ contains also the generator
$u^2$ expressed in terms of  the $\theta$-fremions in~\eqref{uu-theta}. This generator acts
on the fermion generators as
\begin{equation}
u^2 f_j^{(\dagger)}u^{-2}=f_{j+2}^{(\dagger)}
\end{equation}
while it does not leave the generators 
$\f$ and $\e$ invariant, it does leave $\LQGodd$ (together with $\f^L$
and $\e^L$) invariant\footnote{We could equivalently consider a
representation depending on a phase $\varphi$: $u^2 f_j
u^{-2}=e^{i\varphi}f_{j+2}$ and $u^2
f_j^{\dagger}u^{-2}=e^{-i\varphi}f_{j+2}^{\dagger}$, with
$\varphi=2\pi n/L$, $n\in\oZ$, but it has the same centralizer as the
representation with $\varphi=0$.}.  This can be seen by using
fermionic expressions in Sec.~\ref{sec:qg-res} and Sec.~\ref{sec:ferm-exp-cent}.
This observation together with Cor.~\ref{cor:TLa-centralizer} finally prove Thm.~\ref{Thm:centr-JTL-main}.

\section*{Appendix C: Projective $\LQG$-modules $\PP_{1,r}$}
\renewcommand\thesection{C}
\renewcommand{\theequation}{C\arabic{equation}}
\setcounter{equation}{0}

Here, we recall~\cite{BFGT} $\LQG$-action (for $\q=i$) in projective
modules $\PP_{1,r}$, for $r\in\oN$.  Their simple subquotients are
$r$-dimensional irreducible modules $\XX_{1,r}$ spanned by
$\stprp_{m}$, $0\leq m\leq r{-}1$, with the action\footnote{We
simplify a notation used in~\cite{BFGT} assuming $\XX_{1,r} \equiv
\XX^{\alpha(r)}_{1,r}$ with $\alpha(r)=(-1)^{r-1}$, and the same for $\PP_{1,r}$.  }
\begin{equation}\label{eq:sl-irrep}
\begin{split}
&\E\, \stprp_{m} = \F\, \stprp_{m} = 0,\qquad \K\, \stprp_{m} = (-1)^{r-1} \stprp_{m},\\
  \h\, \stprp_{m} &=  \half(r-1-2m)\stprp_{m},\quad
  \e\, \stprp_{m} =  m(r-m)\stprp_{m-1},\quad
  \f\, \stprp_{m} =  \stprp_{m+1},
\end{split}
\end{equation}
where we set $\stprp_{-1}
=\stprp_{r}=0$.
For $r=0$, we also set $\XX_{1,0}\equiv 0$.
The subquotient
structure of $\PP_{1,r}$ is then given as 
\begin{equation}\label{schem-proj}
     \xymatrix@C=2pt@R=24pt@M=1pt@W=2pt{&&\\&\PP_{1,r}\quad=\quad&\\&&}
\xymatrix@C=4pt@R=18pt@M=2pt@W=2pt{
           &\XX_{1,r} \ar[dl] \ar[dr]&\\
 	  \XX_{1,r-1} \ar[dr]
 	    &&\XX_{1,r+1}\ar[dl]\\
        &\XX_{1,r}&
    } 
\end{equation}

For $r > 1$, the projective module $\PP_{1,r}$ has the basis
\begin{equation}\label{left-proj-basis-plus}
  \{\toppr_{m},\botpr_{m}\}_{0\le m\le r-1}
  \cup\{\leftpr_{l}\}_{1\le l\le
  r-1}
\cup\{\rightpr_{l}\}_{0\le l\le
  r},
\end{equation}
where $\{\toppr_{m}\}_{0\le m\le
    r-1}$ is the basis
corresponding to the top module in~\eqref{schem-proj},
$\{\botpr_{m}\}_{0\le m\le r-1}$
to the bottom, $\{\leftpr_{l}\}_{1\le l\le r-1}$ to the left, and
$\{\rightpr_l\}_{0\le l\le r}$ to
the right module. 
For $r=1$, the basis does not contain
$\{\leftpr_{l}\}_{1\le l\le
r-1}$ terms and we imply $\leftpr_{l}~\equiv~0$ in the
action.

We set $\alpha(r)=(-1)^{r-1}$.   The $\LQG$-action on $\PP_{1,r}$ is then given by
\begin{align}
  \K\toppr_{m}&=\alpha(r)\toppr_{m}, \qquad  \K\botpr_{m}=\alpha(r)\botpr_{m}, &\quad 0\le m\le r-1,&\notag\\
  \K\leftpr_{l}&=-\alpha(r)\leftpr_{l}, &\quad 1\le l\le r-1,&\notag\\
  \K\rightpr_{l}&=-\alpha(r)\rightpr_{l}, &\quad 0\le l\le r,&\notag\\
  \E\toppr_{m}&=
    \alpha(r) \ffrac{r-m}{r}\rightpr_{m} + \alpha(r) \ffrac{m}{r}\leftpr_{m},   \qquad \E\botpr_{m}=  0,  &\quad 0\le m\le r-1,&\notag\\
  \E\leftpr_{l}&=
    \alpha(r) (l-r)\botpr_{l-1},  &\quad 1\le l\le r-1,&\label{proj-mod-QGbase}\\
  \E\rightpr_{l}&=
    \alpha(r) l\botpr_{l-1},
&  \quad 0\le l\le r,&\notag\\
  \F\toppr_{m}&=
    \ffrac{1}{r}\rightpr_{m+1}-\ffrac{1}{r}\leftpr_{m+1}, \qquad \F\botpr_{m}= 0 &\quad 0\le m\le r-1, 
    \quad(\leftpr_{r}\equiv0),&\notag\\
  \F\leftpr_{l}&=
    \botpr_{l}, &
  \quad 1\le l\le r-1,&\notag\\
  \F\rightpr_{l}&=
    \botpr_{l}, &
  \quad 0\le l\le r.&\notag
\end{align}

In thus introduced basis, the $s\ell(2)$-generators $\e$, $\f$ and $\h$
act in $\PP_{1,r}$ as in the direct sum
$\XX_{1,r}\oplus\XX_{1,r-1}\oplus\XX_{1,r+1} \oplus \XX_{1,r}$ with
the action defined in~\eqref{eq:sl-irrep}.

\end{document}